%% file: ss-main.tex
\documentclass{article}
\usepackage[utf8]{inputenc}
\usepackage{fullpage}
\usepackage{amsmath, amsfonts, amssymb, stmaryrd}
\usepackage{xspace}
\usepackage{bbm}
\usepackage{hyperref} 
\usepackage{color}    
\usepackage{algpseudocode}
\usepackage[ruled,lined,boxed,commentsnumbered]{algorithm2e}

\newcounter{remarks}

\newtheorem{remark}[remarks]{Remark}
\newcommand{\QED}{\hfill $\square$}

\input{ss-macros}

\title{Sustained Space Complexity}
\author{Jo\"el Alwen\thanks{IST Austria} \and Jeremiah Blocki\thanks{Purdue University} \and Krzysztof Pietrzak\footnotemark[1]}

\begin{document}

\maketitle

\begin{abstract}
\input{ss-abstract}
\end{abstract}

\section{Introduction}\seclab{intro}
\input{ss-intro}

\section{Preliminaries}\seclab{prelim}
\input{ss-preliminaries}

\section{A Graph with Optimal Sustained Space Complexity}\seclab{graph}
\input{ss-graph}

\section{Better Depth-Robustness}\seclab{dr}
\input{ss-dr}

\section{A Pebbling Reduction for Sustained Space Complexity}\seclab{pebbling}
\input{ss-pebbling}

\section{Open Questions}
\input{ss-open}

\section*{Acknowledgments}
The first and third authors were supported by the European Research Council, ERC consolidator grant (682815 - TOCNeT).

\def\shortbib{0}
\bibliographystyle{alpha}
\bibliography{bounded-parallel-mhf,password,abbrev3,../../crypto}

 \appendix
\section{Missing Proofs} \applab{MissingProof}
\input{ss-missing}

\end{document}

%% file: ss-macros.tex
\newcommand{\mdef}[1]{{\ensuremath{#1}}\xspace}  
\newcommand{\myset}[1]{\mdef{\mathbb{#1}}}       
\newcommand{\myalg}[1]{\mdef{\mathcal{#1}}}       
\newcommand{\myfunc}[1]{\mdef{\mathsf{#1}}}      
\newcommand{\myvec}[1]{\mathbf{#1}}               
\newcommand{\myconverter}[1]   {{#1}}
\def \ds {\displaystyle}                         

\DeclareMathOperator*{\polylog}{polylog}

\newcommand{\superscript}[1]{\ensuremath{^{\mbox{\tiny{\textit{#1}}}}}\xspace}
\def \th {\superscript{th}}     

\def \N        {\mdef{\mathbb{N}}}                   
\def \Z        {\mdef{\mathbb{Z}}}                   
\def \R        {\mdef{\mathbb{R}}}                   
\def \from     {\mdef{\leftarrow}}                   

\def \l        {\mdef{\lambda}}                         
\def \negl     {\mdef{\myfunc{negl}}}                
\def \polylog  {\mdef{\myfunc{polylog}}}             
\newcommand{\abs}[1]{\mdef{\left|#1\right|}}         
\newcommand{\flr}[1]{\mdef{\left\lfloor#1\right\rfloor}}              

\newcommand{\node}[2]{\mdef{\langle#1, #2\rangle}}   
\def \G        {\mdef{\mathbb{G}}}

\def \indeg    {\mdef{\mathsf{indeg}}}

\def \sinks    {\mdef{\mathsf{sinks}}}
\def \len      {\mdef{\mathsf{length}}}
\def \depth    {\mdef{\mathsf{depth}}}

\def \parents  {\mdef{\mathsf{parents}}}

\def \ancestors{\mdef{\mathsf{ancestors}}}

\renewcommand{\-}{\mbox{-}}     
\def \seq      {\mdef{\mathsf{seq}}}

\def \smc      {\mdef{\mathsf{smc}}}      

\def \naive    {\mdef{\myalg{N}}}

\def \a        {\mdef{\alpha}}                       
\def \b        {\mdef{\beta}}                        
\def \d        {\mdef{\delta}}                       
\def \e        {\mdef{\epsilon}}                     
\def \eps      {\mdef{\epsilon}}                     
\def \s        {\mdef{\sigma}}                       
\def \t        {\mdef{\tau}}                         
\def \sfin     {\mdef{\mathsf{in}}}
\def \sfout    {\mdef{\mathsf{out}}}
\def \lab      {\mdef{\bf{\mathsf{\lowercase{lab}}}}}    
\def \bfx      {\mdef{{\bf x}}}                      
\def \bs         {\mdef{\bar{\s}}}                         
\def \bfs        {\mdef{\bf{s}}}                           
\def \bfq        {\mdef{\bf{q}}}                           
\newcommand\ares[1][\relax]{\mathcal S^{\textsc{#1}}}
\newcommand\promres[1][\relax]{\ares[$w$-prom]_{#1}}
\newcommand\sromres[1][\relax]{\ares[$w$-srom]_{#1}}
\def \prompar    {\mdef{\Pc^{\textsc{prom}}}}              
\newcommand{\getsr}{{\:{\leftarrow{\hspace*{-3pt}\raisebox{.75pt}{$\scriptscriptstyle\$$}}}\:}}
\def \out        {\mdef{\bf{\mathsf{\lowercase{out}}}}}    
\def \inp        {\mdef{\bf{\mathsf{\lowercase{in}}}}}      
\def \coins      {\mdef{\$}}                               
\def \Pc         {\mdef{\mathbb{P}}}                       
\def \lv                       {\mdef{\myvec{l}}}        
\def \rv                       {\mdef{\myvec{r}}}        

\def \A        {\mdef{\myalg{A}}}                    
\def \dist     {\mdef{\myalg{D}}}                    
\def \Hc       {\mdef{\myset{H}}}                    
\newcommand\indist[1][\relax]{\approx_{#1}}
\def \simul                    {\mdef{\myconverter{\sigma}}}  

\newcommand{\namedref}[2]{\hyperref[#2]{#1~\ref*{#2}}}
\newcommand{\thmlab}[1]{\label{thm:#1}}
\newcommand{\thmref}[1]{\namedref{Theorem}{thm:#1}}
\newcommand{\lemmlab}[1]{\label{lemm:#1}}
\newcommand{\lemmref}[1]{\namedref{Lemma}{lemm:#1}}
\newcommand{\claimlab}[1]{\label{claim:#1}}
\newcommand{\claimref}[1]{\namedref{Claim}{claim:#1}}

\newcommand{\seclab}[1]{\label{sec:#1}}
\newcommand{\secref}[1]{\namedref{Section}{sec:#1}}
\newcommand{\applab}[1]{\label{app:#1}}
\newcommand{\appref}[1]{\namedref{Appendix}{app:#1}}

\newcommand{\deflab}[1]{\label{def:#1}}
\newcommand{\defref}[1]{\namedref{Definition}{def:#1}}
\newcommand{\eqnlab}[1]{\label{eq:#1}}
\newcommand{\eqnref}[1]{\namedref{Equation}{eq:#1}}

\newcounter{theorem}
\newtheorem{theorem}{Theorem}[section]

\newtheorem{lemma}[theorem]{Lemma}
\newtheorem{claim}[theorem]{Claim}

\newtheorem{definition}[theorem]{Definition}

\newenvironment{proof}{\noindent {\sc Proof. }}{\hfill$\Box$\vskip0.2cm}

\newenvironment{proofof}[1]{\noindent {\em Proof of #1.  }}{}

\newenvironment{remindertheorem}[1]{\medskip \noindent {\bf Reminder of  #1.  }\em}{}

\newenvironment{reminderlemma}[1]{\medskip \noindent {\bf Reminder of  #1.  }\em}{}

\newcommand{\peb}{\Pi} 
\newcommand{\Peb}{{\cal P}} 

\newcommand{\ppeb}{\peb^{\parallel}} 
\newcommand{\pPeb}{\Peb^{\parallel}}

\newcommand{\Cspace}{s}
\newcommand{\Ctime}{t}
\newcommand{\Ccc}{{cc}}
\newcommand{\Css}{{ss}}
\newcommand{\Cbm}{{bm}}
\newcommand{\Cspacetime}{{st}}
\newcommand{\Cevery}{\{\Cspace,\Ctime,\Cspacetime,\Ccc\}}

\def \pcc {\ppeb_\Ccc} 

%% file: ss-abstract.tex
%
Memory-hard functions (MHF) are functions whose evaluation cost is dominated by memory cost.  
MHFs are egalitarian, in the sense that evaluating them on dedicated hardware (like FPGAs or ASICs) is not much cheaper than on 
 off-the-shelf hardware (like x86 CPUs).  MHFs have interesting cryptographic applications, most notably to password hashing and securing blockchains.

Alwen and Serbinenko [STOC'15] define the cumulative memory complexity (cmc) of a function as the sum (over all time-steps) of the amount of memory required to compute the function. They advocate that a good MHF must have high cmc. 
 Unlike previous notions, cmc takes into account that dedicated hardware might exploit amortization and parallelism. 
Still, cmc has been critizised as insufficient, as it  fails to capture possible time-memory trade-offs; as memory cost doesn't scale linearly, functions with the same cmc could still have very different actual hardware cost. 

In this work we address this problem, and introduce the notion of sustained-memory complexity, which requires that any algorithm evaluating the function must use a large amount of memory for many steps. We construct functions (in the parallel random oracle model) whose sustained-memory complexity is  almost optimal: our function can be evaluated using $n$ steps and $O(n/\log(n))$ memory, in each step making one query to the (fixed-input length) random oracle, while any algorithm that can make arbitrary many parallel queries to the random oracle, still needs $\Omega(n/\log(n))$ memory for $\Omega(n)$ steps. 

As has been done for various notions (including cmc) before, we reduce the task of constructing an MHFs with high sustained-memory complexity to proving pebbling lower bounds on DAGs. Our main technical contribution is the construction is a family of DAGs on $n$ nodes with constant indegree with high ``sustained-space complexity", meaning that any parallel black-pebbling strategy requires $\Omega(n/\log(n))$ pebbles for at least $\Omega(n)$ steps. 

Along the way we construct a family of maximally ``depth-robust" DAGs with maximum indegree $O(\log n)$, improving upon the construction of Mahmoody et al. [ITCS'13] which had maximum indegree $O\left(\log^2 n \cdot \polylog(\log n)\right)$.


%% file: ss-intro.tex
%
In cryptographic settings we typically consider tasks which can be done efficiently by honest parties, but are infeasible for potential adversaries. This requires an asymmetry in the capabilities of honest and dishonest parties. 
An example are trapdoor functions, where the honest party -- who knows the secret trapdoor key -- can efficiently invert the function, whereas a potential adversary -- who does not have this key -- cannot. 

\subsection{Moderately-Hard Functions}
Moderately hard functions consider a setting where there's no asymmetry, or even worse, the adversary has more capabilities than the honest party. What we want is that the honest party can evaluate the function with some reasonable amount of resources, whereas the adversary should not be able to evaluate the function at significantly lower cost. 
Moderately hard functions have several interesting cryptographic applications, including securing blockchain protocols and for password hashing. 

An early proposal for password hashing is the ``Password Based Key Derivation Function 2" ({\sf PBKDF2})  \cite{kaliski2000pkcs}. This function just iterates a cryptographic hash function like {\sf SHA1} several times ($1024$ is a typical value). Unfortunately {\sf PBKDF2} doesn't make for a good moderately hard function, as evaluating a cryptographic hash function on dedicated hardware like ASCIs (Application Specific Integrated Circuits) can be by several orders of magnitude cheaper in terms of hardware and energy cost than evaluating it on a standard x86 CPU. There have been several suggestions how to construct better, i.e., more ``egalitarian'', moderately hard functions. We discuss the most prominent suggestions below.
\vspace{-.1cm}
\paragraph{Memory-Bound Functions}
Abadi et al. \cite{NDSS:AbadiBW03} observe that the time required to evaluate a function is dominated by the number of cache-misses, and these slow down the computation by about the same time over different architectures. They propose \emph{memory-bound} functions, which are functions that will incur many expensive cache-misses (assuming the cache is not too big). 
They propose a construction which is not very practical as it requires a fairly large (larger than the cache size) incompressible string as input. Their function is then basically pointer jumping on this string. In subsequent work~\cite{DGN03} it was shown that this string can also be locally generated from a short seed.
\vspace{-.1cm}
\paragraph{Bandwidth-Hard Functions}
Recently Ren and Devadas \cite{cryptoeprint:2017:225} suggest the notion of \emph{bandwidth-hard} functions, which is a refinement of memory-bound functions. A major difference being that in their model computation is not completely free, and this assumption -- which of course is satisfied in practice -- allows for much more practical solutions. They also don't argue about evaluation time as~\cite{NDSS:AbadiBW03}, but rather the more important energy cost; the energy spend for evaluating a 
function consists of energy required for on chip computation and memory accesses, only the latter is similar on various platforms. In a bandwidth-hard function the memory accesses dominate the energy cost on a standard CPU, and thus the function cannot be evaluated at much lower energy cost on an ASICs as on a standard CPU.
 \vspace{-.1cm}
\paragraph{Memory-Hard Function}
Whereas memory-bound and bandwidth-hard functions aim at being egalitarian in terms of time and energy, memory-hard functions (MHF), proposed by Percival~\cite{Per09}, aim at being egalitarian in terms of hardware cost. 
A memory-hard function, in his definition, is one where the memory used by the algorithm, multiplied by the amount of time, is high, i.e., it has high space-time (ST) complexity. Moreover, parallelism should not help to evaluate this function at significantly lower cost by this measure. 
The rationale here is that the hardware cost for evaluating an MHF is dominated by the memory cost, and as memory cost does not vary much over different architectures, the hardware cost for evaluating MHFs is not much lower on ASICs than on standard CPUs.

\vspace{-.1cm}
\paragraph{Cumulative Memory Complexity}
Alwen and Serbinenko~\cite{AS15} observe that ST complexity misses a crucial point, amortization. A function might 
have high ST complexity because at some point during the evaluation the space requirement is high, but for most of the time a small memory is sufficient. As a consequence, ST complexity is not multiplicative: a function can have ST complexity $C$, but evaluating $X$ instances of the function can be done with ST complexity much less than $X\cdot C$, so the amortized ST cost is much less than $C$. Alwen and Blocki~\cite{AB16,AB16b} later showed that prominent MHF candidates such as Argon2i~\cite{Argon2}, winner of the Password Hashing Competition~\cite{PHC} do not have high amortized ST complexity. 

To address this issue, \cite{AS15} put forward the notion of cumulative-memory complexity (cmc). The cmc of a function is 
the sum -- over all time steps -- of the memory required to compute the function by any algorithm. Unlike ST complexity, cmc is multiplicative.
\vspace{-.1cm}
\paragraph{Sustained-Memory Complexity} 
Although cmc takes into account amortization and parallelism, it has been observed~(e.g., \cite{TCC:RenDev16,CFRGmailinglistBill} that it still is not sufficient to guarantee egalitarian hardware cost. The reason is simple: if a function has cmc $C$, this could mean that the algorithm minimizing cmc uses some $T$ time steps  and $C/T$ memory on average, but it could also mean it uses time $100\cdot T$ and $C/100\cdot T$ memory on average. In practice this can makes a huge difference because \emph{memory cost doesn't scale linearly}. The length of the wiring required to access memory of size $M$ grows like $\sqrt{M}$ (assuming a two dimensional layout of the circuit). This means for one thing, that -- as we increase $M$ -- the latency of accessing the memory will grow as $\sqrt{M}$, and moreover the space for the wiring required to access the memory will grow like $M^{1.5}$. 

The exact behaviour of the hardware cost as the memory grows is not crucial here, just the point that cmc misses to take into account that it's not linear. In this work we introduce the notion of sustained-memory complexity, which takes this into account. Ideally, we want a function which can be evaluated by a ``na\"ive" sequential algorithm (the one used by the honest parties) in time $T$ using a memory of size $S$ where (1) $S$ should be close to $T$ and (2) any \emph{parallel} algorithm evaluating the function must use memory $S'$ for at least $T'$ steps, where $T'$ and $S'$ should be not much smaller than $T$ and $S$, respectively. 

Property (1) is required so the memory cost dominates the evaluation cost already for small values of $T$. Property (2) means that even a parallel algorithm will not be able to evaluate the function at much lower cost; any parallel algorithm must  make almost as many steps as the na\"ive algorithm during which the required memory is almost as large as the maximum memory $S$ used by the na\"ive algorithm. So, the cost of the best parallel algorithm is similar to the cost of the na\"ive sequential one, even if we don't charge the parallel algorithm anything for all the steps where the memory is below $S'$.

Ren and Devadas~\cite{TCC:RenDev16} previously proposed the notion of ``consistent memory hardness'' which requires that any {\em sequential} evaluation algorithm must either use space $S'$ for at least $T'$ steps, or the algorithm must run for a long time e.g., $T \gg n^2$. Our notion of sustained-memory complexity strengthens this notion in that we consider {\em parallel} evaluation algorithms, and our guarantees are absolute e.g., even if a parallel attacker runs for a very long time he must still use memory $S'$ for at least $T'$ steps. 

In this work we show that functions with almost optimal sustained-memory complexity exist in the random oracle model. 
We note that we must make some idealized assumption on our building block, like being a random oracle, as with the current state of complexity theory, we cannot even prove superlinear circuit lower-bounds for problems in {$\cal NP$}. 
For a given time $T$, our function uses maximal space $S=\Omega(T)$ for the na\"ive algorithm,\footnote{Recall that the na\"ive algorithm is sequential, so $S$ must be in $O(T)$ as in time $T$ the algorithm cannot even touch more than $O(T)$ memory.} while any \emph{parallel} algorithm must make at least $T'=\Omega(T)$ steps during which it uses $S'=\Omega(T/\log(T))=\Omega(S/\log(S))$ of memory.
\vspace{-.1cm}
\paragraph{Graph Labelling}
The functions we construct are defined by DAGs. For a DAG $G_n=(V,E)$, 
we order the vertices $V=\{v_1,\ldots,v_n\}$ in some topological order (so if there's a path from $i$ to $j$ then $i<j$), with 
$v_1$ being the unique source, and $v_n$ the unique sink of the graph. The function is now defined by $G_n$ and the input specifies a random oracle $H$. The output is the label $\ell_n$ of the sink, where the label of a node $v_i$ is recursively defined as 
$\ell_i=H(i,\ell_{p_1},\ldots,\ell_{p_d})$ where $v_{p_1},\ldots,v_{p_d}$ are the parents of $v_i$. 
\paragraph{Pebbling}
Like many previous works, including~\cite{NDSS:AbadiBW03,cryptoeprint:2017:225,AS15} discussed above, 
we reduce the task of proving lower bounds -- in our case, on sustained memory complexity -- for functions as just described, to proving lower bounds on some complexity of a pebbling game played on the underlying graph.

For example \cite{cryptoeprint:2017:225} define a cost function for the so called reb-blue pebbling game, which then implies lower bounds on the bandwidth hardness of the function defined over the corresponding DAG.

Most closely related to this work is~\cite{AS15}, who show that a lower bound the so called sequential (or parallel) \emph{cumulative (black)  pebbling complexity} (cpc) of a DAG implies a lower bound on the sequential (or parallel) cumulative memory complexity (cmc) of the labelling function defined over this graph. Recently~\cite{ABP17} constructed a constant indegree family of DAGs with parallel cpc $\Omega(n^2/\log(n))$, which is optimal~\cite{AB16}, and thus gives functions with optimal cmc. 

The black pebbling game -- as considered in cpc -- goes back to~\cite{HP70,Coo73}. It is defined over a DAG $G=(V,E)$ and goes in round as follows. Initially all nodes are empty. In every round, the player can put a pebble on a node if all its parents contain pebbles (arbitrary many pebbles per round in the parallel game, just one in the sequential). Pebbles can be removed at any time. The game ends when a pebble is put on the sink.  The cpc of such a game is the sum, over all time steps, of the pebbles placed on the graph. The sequential (or parallel) cpc of $G$ is the cpc of the sequential (or parallel) black pebbling strategy which minimizes this cost. 

It's not hard to see that the sequential/parallel cpc of $G$ directly implies the same upper bound on the sequential/parallel cmc of the graph labelling function, as to compute the function in the sequential/parallel random oracle model, one simply mimics the pebbling game, where putting a pebble on vertex $v_i$ with parents $v_{p_1},\ldots,v_{p_d}$ corresponds to the 
query $\ell_i\gets H(i,\ell_{p_1},\ldots,\ell_{p_d})$. And where one keeps a label $\ell_j$ in memory, as long as $v_j$ is pebbled.  
If the labels $\ell_i\in\{0,1\}^w$ are $w$ bits long, a cpc of $p$ translates to cmc of $p\cdot w$.

More interestingly, the same has been shown to hold for interesting notions also for lower bounds. In particular, the \emph{ex-post facto} argument~\cite{AS15} shows that any adversary who computes the label $\ell_n$ with high probability (over the choice of the random oracle $H$) with cmc of $m$, translates into a black pebbling strategy of the underlying graph with cpc almost $m/w$.

In this work we define the \emph{sustained-space complexity} (ssc) of a sequential/parallel black pebbling game, and show that lower bounds on ssc translate to lower bounds on the sustained-memory complexity (smc) of the graph labelling function in the sequential/parallel random oracle model.

Consider a sequential (or parallel) black pebbling strategy (i.e., a valid sequence pebbling configurations where the last configuration contains the sink) for a DAG $G_n=(V,E)$ on $|V|=n$ vertices. 
For some space parameter $ s \le n$, the $s$-ssc of this strategy is the number of pebbling configurations of size at least $s$. The sequential (or parallel) $s$-ssc of $G$ is the strategy minimizing this value. For example, if it's possible to pebble $G$ using $s'<s$ pebbles (using arbitrary many steps), then its $s$-ssc is $0$. Similarly as for csc vs cmc, an upper bound on $s$-ssc implies the same upper bound for $(w\cdot s)$-smc.  In \secref{pebbling} we prove that also lower bounds on ssc translate to lower bounds on smc. 

Thus, to construct a function with high parallel smc, it suffices to construct a family of DAGs with constant indegree and high parallel ssc. In \secref{graph} we construct such a family $\{G_n\}_{n\in\mathbb{N}}$ of DAGs where $G_n$ has $n$ vertices and has indegree $2$, where $\Omega(n/\log(n))$-ssc is in $\Omega(n)$. This is basically the best we can hope for, as our bound on ssc trivially implies a $\Omega(n^2/\log(n))$ bound on csc, which is optimal for any constant indegree graph~\cite{AS15}.
\subsection{High Level Description of our Construction and Proof}
Our construction of a family  $\{G_n\}_{n\in\mathbb{N}}$ of  DAGs with optimal ssc involves three building blocks: 

The first building block is a construction of Paul et al.~\cite{PTC76} of a family of DAGs $\{G_n\}_{n\in\mathbb{N}}$ with $\indeg(G_n)=2$ and space complexity $\Omega(n/\log n)$. More significantly for us they proved that for any sequential pebbling of $G_n$ there is a time interval $[i,j]$ during which at least $\Omega(n/\log n)$ new pebbles are placed on sources of $G_n$ and at least $\Omega(n/\log n)$ are always on the DAG. We extend the proof of Paul et al.~\cite{PTC76} to show that the same holds for any {\em parallel} pebbling of $G_n$. We can argue that $j-i =\Omega(n/\log n)$ for any sequential pebbling since it takes at least this many steps to place $\Omega(n/\log n)$ new pebbles on $G_n$. However, we stress that this argument does not apply to parallel pebblings so this does not directly imply anything about sustained space complexity for parallel pebblings. 

To address this issue we introduce our second building block: a family of $\{G_n^\epsilon\}_{n\in\mathbb{N}}$ of extremely depth robust DAGs with $\indeg(G_n)=O\left(\log n\right)$ --- for any constant $\epsilon >0$ the DAG $G_n^\epsilon$ is $(e,d)$-depth robust for any $e+d \leq (1-\epsilon)n$. We remark that our result improves upon  the construction of Mahmoody et al.\cite{MMV13} whose construction required $\indeg(G_n) = O\left(\log^2 n \polylog(\log n)\right)$ and may be of independent interest (e.g., our construction immediately yields a more efficient construction of proofs of sequential work~\cite{MMV13}). Our construction of $G_n^\epsilon$ is (essentially) the same as Erdos et al.~\cite{EGS75} albeit with much tighter analysis. By overlaying an extremely depth-robust DAG $G_n^\epsilon$ on top of the sources of $G_n$, the construction of Paul et al.~\cite{PTC76}. We can ensure that it takes $\Omega(n/\log n)$ steps to pebble $\Omega(n/\log n)$ sources of $G_n$. However, the resulting graph would have $\indeg(G_n) = O(\log n)$ and would have sustained space $\Omega(n/\log n)$ for at most $O(n/\log n)$ steps. By contrast, we want a $n$-node DAG $G$ with $\indeg(G)=2$ which requires space $\Omega(n/\log n)$ for at least $\Omega(n)$ steps.
 
 Our final tool is to apply the indegree reduction lemma of~Alwen et al.~\cite{ABP17} to $\{G_t^{\epsilon}\}_{t\in\mathbb{N}}$ to obtain a family of DAGs $\{D_t^{\epsilon}\}_{t\in\mathbb{N}}$ such that $D_t$ has $\indeg\left(D_t^{\epsilon}\right)=2$ and $2t \indeg(G_t) = O(t \log t)$ nodes. Each node in $G_t$ is associated with a path of length $2 \indeg(G_t)$ in $D_t^{\epsilon}$ and each path $p$ in $G_t$ corresponds to a path $p'$ of length $|p'|\geq |p| \indeg(G_t)$ in $D_t^{\epsilon}$.  We can then overlay the DAG $D_t^{\epsilon}$ on top of the sources in $G_n$ where $t=\Omega(n/\log n)$ is the number of sources in $G_n$. The final DAG has size $O(n)$ and we can then show that any legal parallel pebbling requires $\Omega(n)$  steps with at least $\Omega(n/\log n)$ pebbles on the DAG.

%% file: ss-preliminaries.tex
%
In this section we introduce common notation, definitions and results from other work which we will be using. In particular the following borrows heavily from~\cite{ABP17,AT17}.

\subsection{Notation}
We start with some common notation. Let $\N = \{0, 1, 2,\ldots\}$, $\N^+ = \{1, 2,\ldots\}$, and $\N_{\ge c} = \{c, c+1, c+2, \ldots\}$ for $c \in \N$. Further, we write $[c] := \{1, 2,\ldots,c\}$ and $[b,c]= \{b, b+1, \ldots, c\}$ where $c \ge b \in \N$. We denote the cardinality of a set $B$ by $|B|$. 

\subsection{Graphs}
The central object of interest in this work are directed acyclic graphs (DAGs). A DAG $G=(V,E)$ has {\em size} $n = |V|$. The indegree of node $v\in V$ is $\d = \indeg(v)$ if there exist $\d$ incoming edges $\d = |(V \times \{v\}) \cap E|$. More generally, we say that $G$ has indegree $\d = \indeg(G)$ if the maximum indegree of any node of $G$ is $\d$. If $\indeg(v) = 0$ then $v$ is called a {\em source} node and if $v$ has no outgoing edges it is called a {\em sink}. We use $\parents_G(v) = \{u \in V: (u,v) \in E\}$ to denote the parents of a node $v \in V$. In general, we use $\ancestors_G(v) := \bigcup_{i \geq 1} \parents_G^i(v)$ to denote the set of all ancestors of $v$ --- here, $\parents_G^2(v) := \parents_G\left(\parents_G(v) \right)$ denotes the grandparents of $v$ and $\parents^{i+1}_G(v) := \parents_G\left( \parents^i_G(v)\right)$. When $G$ is clear from context we will simply write $\parents$ ($\ancestors$).  We denote the set of all sinks of $G$ with $\sinks(G)=\{v\in V: \nexists (v,u)\in E\}$ --- note that $\ancestors\left(\sinks(G) \right) = V$. The length of a directed path $p = (v_1, v_2, \ldots, v_z)$ in $G$ is the number of nodes it traverses $\len(p):=z$. The depth $d=\depth(G)$ of DAG $G$ is the length of the longest directed path in $G$. We often consider the set of all DAGs of fixed size $n$ $\G_n := \{G=(V,E)\ : \ |V|=n\}$ and the subset of those DAGs at most some fixed indegree $\G_{n,\d} := \{G\in\G_n :  \indeg(G) \le \d\}$. Finally, we denote the graph obtained from $G=(V,E)$ by removing nodes $S\subseteq V$ (and incident edges) by $G-S$ and we denote by $G[S]=G-(V\setminus S)$ the graph obtained by removing nodes $V \setminus S$ (and incident edges).

The following is an important combinatorial property of a DAG for this work.

\begin{definition}[Depth-Robustness]
 For $n\in\N$ and $e,d \in [n]$ a DAG $G=(V,E)$ is {$(e,d)$-depth-robust} if
 $$\forall S \subset V~~~|S| \le e \Rightarrow \depth(G-S) \ge d.$$
\end{definition}

The following lemma due to Alwen et al.~\cite{ABP17} will be useful in our analysis. Since our statement of the result is slightly different from~\cite{ABP17} we include a proof in  \appref{MissingProof} for completeness.
\newcommand{\indegReductionLemma}{({\bf Indegree-Reduction}) Let $G=(V=[n],E)$ be a $(e,d)$-depth robust DAG on $n$ nodes and let $\d = \indeg(G)$. We can efficiently construct a  DAG $G' = (V'=[2n\d],E')$ on $2n\d$ nodes with $\indeg(G')=2$ such that for each path $p=(x_1,...,x_k)$ in $G$ there exists a corresponding path $p'$ of length $\geq k\d$ in $G'\left[\bigcup_{i=1}^k  [2(x_i-1)\d+1,2x_i\d]\right]$ such that $2x_i \d \in p'$ for each $i \in [k]$.  In particular, $G'$ is $(e,d\d)$-depth robust.}
\begin{lemma}\cite[Lemma 1]{ABP17}  \lemmlab{IndegRed}
\indegReductionLemma
\end{lemma}

\subsection{Pebbling Models}
The main computational models of interest in this work are the parallel (and sequential) pebbling games played over a directed acyclic graph. Below we define these models and associated notation and complexity measures. Much of the notation is taken from~\cite{AS15,ABP17}.

\begin{definition}[Parallel/Sequential Graph Pebbling]\deflab{pebbling}
Let $G= (V,E)$ be a DAG and let $T \subseteq V$ be a target set of nodes to be pebbled. 
A {\em pebbling configuration (of $G$)} is a subset $P_i\subseteq V$. A legal \emph{parallel} pebbling of $T$ is a sequence $P=(P_0,\ldots,P_t)$ of {\em pebbling configurations} of $G$ where $P_0 = \emptyset$ and which satisfies conditions 1 \& 2 below. A \emph{sequential} pebbling additionally must satisfy condition 3.
 \begin{enumerate}
  \item At some step every target node is pebbled (though not necessarily simultaneously).
    $$ \forall x \in T~ \exists z \le t ~~:~~x\in P_z.$$
  \item A pebble can be added only if all its parents were pebbled at the end of the previous step.
    $$ \forall i \in [t]~~:~~ x \in (P_i \setminus P_{i-1}) ~\Rightarrow~ \parents(x) \subseteq P_{i-1}.$$
  \item At most one pebble is placed per step.
    $$ \forall i \in [t]~~:~~ |P_i \setminus P_{i-1}|\le 1 \ .$$
 \end{enumerate}
We denote with $\Peb_{G,T}$ and $\pPeb_{G,T}$ the set of all legal sequential and parallel pebblings of $G$ with target set $T$, respectively. Note that 
$\Peb_{G,T}\subseteq \pPeb_{G,T}$. We will mostly be interested in the case where $T = \sinks(G)$ in which case we write $\Peb_{G}$ and $\pPeb_{G}$.
\end{definition}

\begin{definition}[Pebbling Complexity]
The standard notions of {\em time}, {\em space}, {\em space-time} and {\em cumulative (pebbling) complexity} (cc) of a pebbling $P=\{P_0,\ldots,P_t\}\in\pPeb_G$ are defined to be:
$$
\peb_\Ctime(P)=t ~~~~~
\peb_\Cspace(P)= \max_{i\in[t]}|P_i| ~~~~~
\peb_\Cspacetime(P)= \peb_\Ctime(P)\cdot\peb_\Cspace(P) ~~~~~
\peb_\Ccc(P)= \sum_{i\in [t]}|P_i| \ .
$$
For $\alpha\in\Cevery$ and a target set $T \subseteq V$, the sequential and parallel pebbling complexities of $G$ are defined as
$$
\peb_\alpha(G,T)=\min_{P\in\Peb_{G,T}}\peb_\alpha(P) \qquad\textrm{and}\qquad
\ppeb_\alpha(G,T)=\min_{P\in\pPeb_{G,T}}\peb_\alpha(P) \ .
$$
When $T = \sinks(G)$ we simplify notation and write $\peb_\alpha(G)$ and $\ppeb_\alpha(G)$.
\end{definition}

\noindent The following defines a sequential pebbling obtained naturally from a parallel one by adding each new pebble on at a time.

\begin{definition}\deflab{seq}
Given a DAG $G$ and $P = \left(P_0,\ldots,P_t\right) \in \pPeb_G$ the {\em sequential transform} $\seq(P) = P' \in \peb_G$ is defined as follows: Let difference $D_j = P_i \setminus P_{i-1}$ and let $a_i = \left| P_i \setminus P_{i-1}\right|$ be the number of new pebbles placed on $G_n$ at time $i$. Finally, let $A_j = \sum_{i=1}^j a_i$ ( $A_0=0$) and let $D_j[k]$ denote the $k\th$ element of $D_j$ (according to some fixed ordering of the nodes). We can construct $P' = \left(P_{1}',\ldots,P_{A_t}'\right) \in \Peb(G_n)$ as follows: (1) $P_{A_i}' = P_i$ for all $i\in[0,t]$, and (2) for $k\in[1, a_{i+1}]$ let $P_{A_i+k}' = P_{A_i+k-1}' \cup D_j[k] $. 
\end{definition}

\noindent If easily follows from the definition that the parallel and sequential space complexities differ by at most a multiplicative factor of $2$.

\begin{lemma}\lemmlab{seqparspace}
 For any DAG $G$ and $P\in \pPeb_G$ it holds that $\seq(P) \in \Peb_G$ and $\peb_\Cspace(\seq(P)) \le 2 *\ppeb_\Cspace(P)$. In particular $\peb_\Cspace(G) /2 \le\ppeb_\Cspace(G)$.
\end{lemma}

\begin{proof}
 Let $P\in\pPeb_G$ and $P' = \seq(P)$. Suppose $P'$ is not a legal pebbling because $v\in V$ was illegally pebbled in $P'_{A_i+k}$. If $k=0$ then $\parents_G(v) \not\subseteq P'_{A_{i-1}+a_i-1}$ which implies that $\parents_G(v) \not\subseteq P_{i-1}$ since $P_{i-1} \subseteq P'_{A_{i-1}+a_i-1}$. Moreover $v\in P_i$ so this would mean that also $P$ illegally pebbles $v$ at time $i$. If instead, $k>1$ then $v\in P_{i+1}$ but since $\parents_G(v) \not\subseteq P'_{A_i+k-1}$ it must be that $\parents_G(v) \not\subseteq  P_i$ so $P$ must have pebbled $v$ illegally at time $i+1$. Either way we reach a contradiction so $P'$ must be a legal pebbling of $G$. To see that $P'$ is complete note that $P_0 = P'_{A_0}$. Moreover for any sink $u \in V$ of $G$ there exists time $i\in[0,t]$ with $u\in P_i$ and so $u\in P'_{A_i}$. Together this implies $P'\in\Peb_G$.

 Finally, it follows by inspection that for all $i\ge 0$ we have $|P'_{A_i}| = |P_i|$ and for all $0<k<a_i$ we have $|P'_{A_i+k}| \le |P_i| + |P_{i+1}|$ which implies that $\peb_\Cspace(P') \le 2*\ppeb_\Cspace(P)$.
\end{proof}

\noindent New to this work is the following notion of sustained-space complexity.

\begin{definition}[Sustained Space Complexity]
 For $s\in \N$ the {\em $s$-sustained-space} ($s$-ss) complexity of a pebbling $P=\{P_0,\ldots,P_t\}\in\pPeb_G$ is:
 $$ \peb_\Css(P,s) = |\{ i \in [t] : |P_i| \ge s\}|.$$
 More generally, the sequential and parallel $s$-sustained space complexities of $G$ are defined as
 $$
  \peb_\Css(G,T,s)=\min_{P\in\Peb_{G,T}}\peb_\Css(P,s) \qquad\textrm{and}\qquad
  \ppeb_\Css(G,T,s)=\min_{P\in\pPeb_{G,T}}\peb_\Css(P,s) \ .
 $$
 As before, when $T = \sinks(G)$ we simplify notation and write $\peb_\Css(G,s)$ and $\ppeb_\Css(G,s)$.
\end{definition}

\begin{remark} (On Amortization)
An astute reader may observe that $\ppeb_\Css$ is not amortizable. In particular, if we let $G^{\bigotimes m}$ denotes the graph which consists of $m$ independent copies of $G$ then we may have  $\ppeb_\Css\left(G^{\bigotimes m},s\right) \ll m \ppeb_\Css(G,s)$. However, we observe that the issue can be easily corrected by defining the {\em amortized $s$-sustained-space} complexity of a pebbling $P=\{P_0,\ldots,P_t\}\in\pPeb_G$:
 $$ \peb_{am,\Css}(P,s) \sum_{i=1}^t \left\lfloor \frac{\left| P_i\right|}{s} \right\rfloor.$$
In this case we have  $\ppeb_{am,\Css}\left(G^{\bigotimes m},s\right)= m \ppeb_{am,\Css}(G,s)$ where $\ppeb_{am,\Css}(G,s) \doteq \min_{P\in\pPeb_{G,\sinks(G)}} \peb_{am,\Css}(P,s)$. We also remark that {\em $s$-sustained-space}  complexity is a strictly stronger guarantee than {\em amortized $s$-sustained-space} since $\ppeb_\Css(G,s) \leq \ppeb_{am,\Css}(G,s)$. Thus, all of our lower bounds from $\ppeb_\Css$ also hold with respect to $\ppeb_{am,\Css}$.
\end{remark}

\noindent The following shows that the indegree of any graph can be reduced down to $2$ with out loosing too much in the parallel sustained space complexity. The technique is similar the indegree reduction for cumulative complexity in~\cite{AS15}. The proof is in \appref{MissingProof}.

\newcommand{\IndegReductionForParallelSustainedSpace}{$$ \forall G \in \G_{n,\d}, ~~ \exists H \in \G_{n',2} \mbox{ such that }\forall s\ge 0~~ \ppeb_\Css(H,s/(\d-1)) = \ppeb_\Css(G,s) \mbox{ where } n' \in [n,\d n].$$}

\begin{lemma}[Indegree Reduction for Parallel Sustained Space]\lemmlab{indegred}
\IndegReductionForParallelSustainedSpace
\end{lemma}

%

%% file: ss-graph.tex
%
In this section we construct and analyse a graph with very high sustained space complexity by modifying the graph of~\cite{PTC76} using the graph of~\cite{EGS75}. \thmref{main}, our main theorem, states that there is a family of constant indegree DAGs $\{G_n\}_{n=1}^\infty$ with maximum possible sustained space $\peb_\Css\left(G_n, \Omega(n/\log n) \right) = \Omega(n)$.

\begin{theorem} \thmlab{main}
For some constants $c_4,c_5 >0$ there is a family of DAGs $\{G_n\}_{n=1}^\infty$ with $\indeg\left(G_n\right) = 2$,  $O(n)$ nodes and $\ppeb_\Css\left(G_n,c_4n/\log n \right) \geq c_5n$.   
\end{theorem}

\begin{remark}
We observe that \thmref{main} is essentially optimal in an asymptotic sense. Hopcroft et al.~\cite{HPV77} showed that any DAG $G_n$ with $\indeg(G_n)=O(1)$ can be pebbled with space at most $\ppeb_\Cspace(G_n) = O\left(n/\log n\right)$. Thus,  $\peb_\Css\left(G_n,s_n=\omega\left(n/\log n\right)\right) = 0$ for any DAG $G_n$ with $\indeg(G_n)=O(1)$ since $s_n > \peb_\Cspace(G_n)$.  \footnote{Furthermore, even if we restrict our attention to pebblings which finish in time $O(n)$ we still have $\peb_\Css\left(G_n,f(n)\right) \leq g(n)$ whenever $f(n)g(n) = \omega\left(\frac{n^2 \log \log n}{\log n}\right)$ and $\indeg(G_n)=O(1)$. In particular, Alwen and Blocki~\cite{AB16} showed that for any $G_n$ with $\indeg(G_n)=O(1)$ then there is a pebbling $P = (P_0,\ldots,P_n) \in \ppeb_{G_n}$ with $\pcc(P) \leq O\left(\frac{n^2 \log \log n}{\log n}\right)$. By contrast, the generic pebbling~\cite{HPV77} of any DAG with $\indeg=O(1)$ in space $O\left(n/\log n\right)$ can take exponentially long. }
\end{remark}

We now overview the key technical ingredients in the proof of \thmref{main}.

\subsubsection*{Technical Ingredient 1: High Space Complexity DAGs} The first key building blocks is a construction of Paul et al.~\cite{PTC76} of a family of $n$ node DAGs $\{G_n\}_{n=1}^\infty$ with space complexity $\peb_\Cspace(G_n)=\Omega(n/\log n)$ and $\indeg(G_n)=2$. \lemmref{seqparspace} implies that $\ppeb_\Cspace(G_n) = \Omega(n/\log n)$ since  $\peb_\Cspace(G_n) /2 \le\ppeb_\Cspace(G_n)$. However, we stress that this does not imply that the sustained space complexity of $G_n$ is large. In fact, by inspection one can easily verify that $\depth(G_n) = O(n/\log n)$ so we have $\peb_\Css(G_n, s) \leq O(n/\log n)$ for any space parameter $s>0$. Nevertheless, one of the core lemmas from \cite{PTC76} will be very useful in our proofs. In particular, $G_n$ contains $O(n/\log n)$ source nodes and \cite{PTC76} proved that for any sequential pebbling $P = (P_0,\ldots,P_t) \in \peb_{G}$ we can find an interval $[i,j] \subseteq [t]$ during which $\Omega(n/\log n)$ sources are (re)pebbled and at least $\Omega(n/\log n)$ pebbles are always on the graph --- see \thmref{PTC} in \appref{MissingProof} for a formal statement of their original result.  

As we show in \thmref{PTCParallel} the same claim holds for all parallel pebblings $P \in \ppeb_{G_n}$. Since  Paul et al.~\cite{PTC76} only considered sequential black pebblings we include the straightforward proof of \thmref{PTCParallel} in \appref{MissingProof} for completeness. Briefly, to prove \thmref{PTCParallel} we simply consider the sequential transform $\seq(P) = (Q_0,\ldots,Q_{t'})   \in \peb_{G_n}$ of the parallel pebbling $P$. Since $\seq(P)$ is sequential we can find an interval $[i',j'] \subseteq [t']$ during which $\Omega(n/\log n)$ sources are (re)pebbled and at least $\Omega(n/\log n)$ pebbles are always on the graph $G_n$. We can then translate $[i',j']$ to a corresponding interval $[i,j] \subseteq [t]$ during which the same properties hold for $P$. 

\newcommand{\thmPTCParallel}{There is a family of DAGs $\{G_n= (V_n=[n],E_n)\}_{n=1}^\infty$ with $\indeg\left(G_n\right) = 2$ with the property that for some positive constants $c_1,c_2,c_3 > 0$ such that for each $n \geq 1$ the set $S = \{v \in [n] ~:~ \parents(v)=\emptyset \}$ of sources of $G_n$ has size  $\left| S\right| \leq c_1n/\log n$ and for any legal pebbling $P=\left(P_1,\ldots,P_t\right) \in \pPeb_{G_n}$ there is an interval $[i,j] \subseteq [t]$ such that  (1) $\left| S \cap \bigcup_{k =i}^j P_k \setminus P_{i-1} \right| \geq c_2n/\log n$ i.e.,  at least $c_2 n/\log n$ nodes in $S$ are (re)pebbled during this interval, and (2) $\forall k \in [i,j], \left| P_k \right| \geq c_3n/\log n$ i.e., at least $c_3n/\log n$ pebbles are always on the graph. }
\begin{theorem} \thmlab{PTCParallel}
\thmPTCParallel
\end{theorem}

One of the key remaining challenges to establishing high sustained space complexity is that the interval $[i,j]$ we obtain from \thmref{PTCParallel} might be very short for parallel black pebblings. For sequential pebblings it would take $\Omega(n/\log n)$ steps to (re)pebble $\Omega(n/\log n)$ source nodes since we can add at most one new pebble in each round. However, for parallel pebblings we cannot rule out the possibility that all $\Omega(n/\log n)$ sources were pebbled in a single step! 

A first attempt at a fix is to modify $G_n$ by overlaying a path of length $\Omega(n)$ on top of these $\Omega(n/\log n)$ source nodes to ensure that the length of the interval $j-i+1$ is sufficiently large. The hope is that it will take now at least $\Omega(n)$ steps to (rep)pebble any subset of $\Omega(n/\log n)$ of the original sources since these nodes will be connected by a path of length $\Omega(n)$. However, we do not know what the pebbling configuration looks like at time $i-1$. In particular, if $P_{i-1}$ contained just $\sqrt{n}$ of the nodes on this path then the it would be possible to (re)pebble all nodes on the path in at most $O\left(\sqrt{n}\right)$ steps. This motivates our second technical ingredient: extremely depth-robust graphs.

\subsubsection*{Technical Ingredient 2: Extremely Depth-Robust Graphs} Our second ingredient is a family $\{D_n^\epsilon\}_{n=1}^\infty$ of highly depth-robust DAGs with $n$ nodes and $\indeg(D_n) = O(\log n)$. In particular, $D_n^\epsilon$ is $(e,d)$-depth robust for {\em any} $e+d \leq n(1-\epsilon)$. We show how to construct such a family $\{D_n^\epsilon\}_{n=1}^\infty$ for for any constant $\epsilon>0$ in \secref{dr}. Assuming for now that such a family exists we can overlay $D_m$ over the $m\leq c_1 n/\log n$ sources of $G_n$. Since $D_m^\epsilon$ is {\em highly} depth-robust it will take at least $c_2n/\log n - \epsilon m \geq c_2 n/\log n - \epsilon c_1 n/\log n = \Omega(n/\log n)$ steps to pebble these $c_2 n/\log n$ sources during the interval $[i,j]$. 

Overlaying $D_m^\epsilon$ over the $m = O(n/\log (n))$ sources of $G_n$ yields a DAG $G$ with $O(n)$ nodes, $\indeg(G)=O(\log n)$ and $\ppeb_\Css\left(G,c_4n/\log n \right) \geq c_5n/\log n$ for some constants $c_4,c_5 >0$. While this is progress it is still a weaker result than \thmref{main} which promised  a DAG $G$ with $O(n)$ nodes, $\indeg(G)=2$ and $\ppeb_\Css\left(G,c_4n/\log n \right) \geq c_5n$ for some constants $c_4,c_5 >0$. Thus, we need to introduce a third technical ingredient: indegree reduction. 

\subsubsection*{Technical Ingredient 3: Indegree Reduction}  To ensure $\indeg(G)=2$ we instead apply indegree reduction algorithm from \lemmref{IndegRed} to $D_m^\epsilon$  to obtain a graph $J_m^\epsilon$ with $2m\d = O(n)$ nodes $[2\d m]$ and $\indeg(J_m^\epsilon)=2$ before overlaying --- here $\d = \indeg(D_m^\epsilon)$. We then associate the $m$ sources of $G_n$ with the nodes $\{2\d v~:v \in [m]\}$ in $J_m^\epsilon$. $J_m^\epsilon$ is $(e,\d d)$-depth robust for any $e+d \leq (1-\epsilon) m$, which seems to suggests that it will take $\Omega(n)$ steps to (re)pebble $c_2 n/\log n$ sources during the interval. However, we still run into the same problem: In particular, suppose that at some point in time $k$ we can find a set $T \subseteq \{2v\d:v \in [m]\}\setminus P_k$ with $|T| \geq c_2 n/\log n$ (e.g., a set of sources in $G_n$) such that the longest path running through $T$ in $J_m^\epsilon - P_{k}$ has length at most $c_5 n$. If the interval $[i,j]$ starts at time $i=k+1$ then cannot ensure that it will take time $\geq c_5 n$ to (re)pebble these $c_2 n/\log n$ source nodes. 

\claimref{LDImpliesSSTwo} addresses this challenge directly. If such a problematic time $k$ exists then  \claimref{LDImpliesSSTwo} implies that we must have $\ppeb_\Css\left(P, \Omega(n/\log n)) \right) \geq \Omega(n)$. At a high level the argument proceeds as follows: suppose that we find such a problem time $k$ along with a set $T \subseteq \{2v\d:v \in [m]\}\setminus P_k$ with $|T| \geq c_2 n/\log n$ such that $\depth\left(J_m^\epsilon[T]\right) \leq c_5 n$. Then for any time $r \in [k-c_5 n,k]$ we know that the the length of the longest path running through $T$ in  $J_m^\epsilon - P_r$ is at most $\depth\left(J_m^\epsilon[T]- P_r\right) \leq c_5 n +(k-r) \leq 2 c_5 n$ since the depth can decrease by at most one each round. We can then use the extreme depth-robustness of $D_m^\epsilon$ and the construction of $J_m^\epsilon$ to argue that $\left| P_r \right| = \Omega(n/\log n)$ for each $r \in [k-c_5 n,k]$. Finally, if no such problem time $k$ exists then the interval $[i,j]$ we obtain from \thmref{PTCParallel} must have length at least $ i- j \geq c_5 n$. In either case we have $\ppeb_\Css\left(P, \Omega(n/\log n)) \right) \geq \Omega(n)$.

\begin{proofof}{\thmref{main}}
We begin with the family of DAGs $\{G_n\}_{n=1}^\infty$ from \thmref{PTCParallel} and \thmref{PTC}. Fixing $G_n=([n],E_n)$ we let $S = \{v \in [n]: \parents(v) = \emptyset\}\subseteq V$ denote the sources of this graph and we let $c_1,c_2,c_3 > 0$ be the constants from \thmref{PTCParallel}. Let $\epsilon \leq c_2/(4c_1)$. By \thmref{ExtremeDepthRobust} we can find a depth-robust DAG $D_{|S|}^\epsilon$ on $|S|$ nodes which is $(a|S|,b|S|)$-DR for any $a+b\leq 1-\epsilon$ with indegree $c' \log n \leq \d = \indeg(D) \leq c'' \log(n)$ for some constants $c',c''$. We let $J_{|S|}^\epsilon$ denote the indegree reduced version of $D_{|S|}^\epsilon$ from \lemmref{IndegRed} with $2|S|\d = O(n)$ nodes and $\indeg=2$. To obtain our DAG $G$ from $J_n^\epsilon$ and $G_n$ we associate each of the $S$ nodes $2v\d$ in $J_n^\epsilon$ with one of the nodes in $S$. We observe that $G$ has at most $2|S|\d+n = O(n)$ nodes and that $\indeg(G) \leq \max\left\{\indeg(G_n),\indeg\left(J_n^\epsilon\right)\right\}=2$ since we do not increase the indegree of any node in $J_n^\epsilon$ when overlaying and in $G_n$ do not increase the indegree of any nodes other that sources $S$ (which may now have indegree $2$ in $J_n^\epsilon$).

Let $P=(P_0,\ldots,P_t) \in \pPeb_{G}$ be given and observe that by restricting $P'_i = P_i \cap V(G_n) \subseteq P_i$ we have a legal pebbling $P'=(P_0',\ldots,P_t') \in \pPeb_{G_n}$ for $G_n$. Thus, by \thmref{PTCParallel} we can find an interval $[i,j]$ during which at least $c_2n/\log n$ nodes in $S$ are (re)pebbled and $\forall k \in [i,j]$ we have $\left|P_k\right| \geq c_3n/\log n$. We use $T=S \cap \bigcup_{x=i}^j P_x - P_{i-1}$ to denote the source nodes of $G_n$ that are (re)pebbled during the interval $[i,j]$. We now set $c_4 = c_2/4$ and $c_5 = c_2 c'/4$ and consider two cases:

{\bf Case 1:} We have $\depth\left(\ancestors_{G-P_i}(T)\right) \geq |T|\d/4$. In other words at time $i$  there is an unpebbled path of length $\geq |T|\d/4$ to some node in $T$. In this case, it will take at least $j-i\geq |T|\d/4$ steps to pebble $T$ so we have $|T|\d/4 = \Omega(n)$ steps with at least $c_3n/\log n$ pebbles. Because $c_5 = c_2 c'/4$ it follows that $|T|\d/4 \geq c_2c' n \geq c_5 n$. Finally, since $c_4 \leq c_2$ we have $\ppeb_\Css\left(G_n,c_4n/\log n \right) \geq c_5n$. \\

{\bf Case 2:}  We have $\depth\left(\ancestors_{G-P_i}(T)\right) < |T|\d/4$. In other words at time $i$ there is no unpebbled path of length $\geq |T|\d/4$ to any node in $T$. Now \claimref{LDImpliesSSTwo} directly implies that $\ppeb_\Css\left(P, |T|-\epsilon |S| - |T|/2) \right) \geq \d |T|/4$. This in turn implies that $\ppeb_\Css\left(P, (c_2/2)n/(\log n) -\epsilon|S| \right) \geq \d c_2n/(2\log n) $. We observe that $\d c_2n/(2\log n) \geq c_5n$ since, we have $c_5 = c_2 c'/4$. We also observe that $(c_2/2 )n/\log n -\epsilon|S| \geq (c_2/2-\epsilon c_1)n/\log n \geq  (c_2/2-c_2/4)n/\log n \geq c_2n/(4 \log n) = c_4n$  since $|S| \leq c_1 n/\log n$, $\epsilon \leq c_2/(4c_1)$ and $c_4=c_2/4$. Thus, in this case we also have $\ppeb_\Css\left(P, c_4n/\log n \right) \geq c_5 n$.\QED
\end{proofof}


\begin{claim} \claimlab{LDImpliesSSTwo}
Let $D_n^\epsilon$ be an DAG with nodes $V\left(D_n^\epsilon \right)=[n]$, indegree $\d = \indeg\left( D_n^\epsilon\right)$ that is $(e,d)$-depth robust for all $e,d>0$ such that $e+d\leq (1-\epsilon)n$, let $J_n^\epsilon$ be the indegree reduced version of $D_n^\epsilon$ from \lemmref{IndegRed} with $2 \d$ nodes and $\indeg\left(J_n^\epsilon\right) = 2$, let $T \subseteq [n]$  and let $P=(P_1,\ldots,P_t) \in \pPeb_{J_n^\epsilon,\emptyset}$ be a (possibly incomplete) pebbling of $J_n^\epsilon$. Suppose that during some round $i$ we have $\depth\left(\ancestors_{J_n^\epsilon-P_i}\left(\bigcup_{v \in T} \{2\d v \}\right)\right) \leq c\d |T|$ for some constant $0 < c < \frac{1}{2}$. Then $\ppeb_\Css\left(P, |T|-\epsilon n - 2c|T|) \right) \geq c \d |T|$.
\end{claim}
\begin{proof}
For each time step $r$ we let $H_r = \ancestors_{J_n^\epsilon-P_r}\left(\bigcup_{v \in T} \{2\d v \}\right) $ and let $k < i$ be the last pebbling step before $i$ during which $\depth(G_k) \geq  2c|T|\d$. Observe that $k-i \geq \depth(H_k) - \depth(H_i) \geq cn\d$ since we can decrease the length of any unpebbled path by at most one in each pebbling round. We also observe that $\depth(H_k) = c|T|\d$ since $\depth(H_k)-1 \leq \depth(H_{k+1}) < 2c|T|\d$. 

Let $r \in [k,i]$ be given then, by definition of $k$, we have $\depth\left(H_r\right) \leq 2c|T|\d$. Let $P_r' = \{v \in V(D_n^\epsilon): P_r \cap [2\d(v-1)+1,2\d v] \neq \emptyset \}$ be the set of nodes $v \in [n]=V\left(D_n^\epsilon\right)$ such that the corresponding path $2\d(v-1)+1,\ldots,2\d v$ in $J_n^\epsilon$ contains at least one pebble at time $r$. By depth-robustness of $D_n^\epsilon$ we have 
\begin{equation} \eqnlab{depthPartOne} \depth\left(D_n^\epsilon[T]  - P_r'\right) \geq |T|-|P_r'|-\epsilon n \ . \end{equation} On the other hand, exploiting the properties of the indegree reduction from \lemmref{IndegRed}, we have 
\begin{equation}\eqnlab{depthPartTwo} \depth\left(D_n^\epsilon[T]  - P_r'\right) \d \leq \depth\left(H_r\right)  \leq 2c|T|\d \ .   \end{equation}
Combining \eqnref{depthPartOne} and \eqnref{depthPartTwo} we have \[|T|-|P_r'|-\epsilon n \leq \depth\left(D_n^\epsilon[T]  - P_r'\right)   \leq 2c|T| \ . \]
It immediately follows that 
$\left| P_r\right| \geq |P_r'| \geq |T|- 2c |T| - \epsilon n$ for each $r \in [k,i]$ and, therefore, $\ppeb_\Css\left(P, |T|-\epsilon n - 2c|T| \right) \geq c \d |T|$.
\end{proof}

\begin{remark} (On the Explicitness of Our Construction) Our construction of a family of DAGs with high sustained space complexity is explicit in the sense that there is a probabilistic polynomial time algorithm which, except with very small probability, outputs an $n$ node DAG $G$ that has high   sustained space complexity. In particular, \thmref{main} relies on an explicit construction of \cite{PTC76}, and the extreme depth-robust DAGs from \thmref{ExtremeDepthRobust}. The construction of \cite{PTC76} in turn uses an object called superconcentrators. Since we have explicit constructions of superconcentrators~\cite{gabber1981explicit} the construction of \cite{PTC76} can be made explicit. While the proof of the existence of a family of extremely depth-robust DAGs is not explicit the proof uses a probabilistic argument and can be adapted to obtain a  probabilistic polynomial time which, except with very small probability, outputs an $n$ node DAG $G$ that is extremely depth-robust.  In practice, however it is also desirable to ensure that there is a local algorithm which, on input $v$, computes the set $\parents(v)$ in time $\polylog(n)$. It is an open question whether any DAG $G$ with high sustained space complexity allows for highly efficient computation of the set $\parents(v)$. 
\end{remark}

%% file: ss-dr.tex
%
In this section we improve on the original analysis of Erdos et al.~\cite{EGS75}, who constructed a family of DAGs with $\indeg=O(\log n)$ that is $\left(e= \Omega(n),d=\Omega(n)\right)$-depth robust. Such a DAG $G_n$ is not sufficient for us since we require that the subgraph $G_n[T]$ is also highly depth robust for any sufficiently large subset $T \subseteq V_n$ of nodes e.g., for any $T$ such that $|T| \geq n/1000$. For any fixed constant $\epsilon > 0$ \cite{MMV13} constructs a family of DAGs $\{G_n^\epsilon\}_{n=1}^\infty$ which is $(\alpha n,\beta n)$-depth robust for any positive constants $\alpha,\beta$ such that  $\alpha + \beta \leq 1-\epsilon$ but their construction has indegree $O\left(\log^2 n \cdot \polylog \left(\log n\right)\right)$. By contrast our results in the previous section assumed the the existence of such a family of DAGs with $\indeg(G_n) = O(\log n)$. 

In fact our family of DAGs is essentially the same as ~\cite{EGS75} with one minor modification to make the construction for for all $n > 0$. Our contribution in this section is an improved analysis which shows that the family of DAGs $\{G_n^\epsilon\}_{n=1}^\infty$ with indegree $O\left(\log n\right)$ is $(\alpha n,\beta n)$-depth robust for any positive constants $\alpha,\beta$ such that  $\alpha + \beta \leq 1-\epsilon$. 

We remark that if we allow our family of DAGs to have $\indeg(G_n) =  O(\log n \log^* n)$ then we can eliminate the dependence on $\epsilon$ entirely. In particular, we can construct a family of DAGs  $\{G_n^\epsilon\}_{n=1}^\infty$ with $\indeg(G_n) =  O(\log n \log^* n)$ such that for any positive constants such that $\alpha + \beta < 1$ the DAG $G_n$ is $(\alpha n,\beta n)$-depth robust for all suitably large $n$.


\begin{theorem} \thmlab{ExtremeDepthRobust}
Fix $\epsilon >0$ then there exists a family of DAGs $\{G_n^\epsilon\}_{n=1}^\infty$ with $\indeg\left(G_n^\epsilon\right)=O(\log n)$ that is $\left( \alpha n, \beta n \right)$-depth robust for any constants $\alpha,\beta$ such that  $\alpha + \beta < 1-\epsilon$. 
\end{theorem}

The proof of \thmref{ExtremeDepthRobust} relies on \lemmref{anyDelta}, \lemmref{GoodNodesConnected} and \lemmref{CountGoodNodes}. We say that $G$ is a $\delta$-local expander if for every node $x \in [n]$ and every $r \leq x, n-x$ and every pair $A \subseteq I_r(x)\doteq \{x-r-1,\ldots,x\}, B \subseteq I_r^*(x) \doteq \{x+1,\ldots,x+r\}$ with size $\left|A\right|, \left|B \right| \geq \delta r$ we have $A \times B \cap E \neq \emptyset$ i.e., there is a directed edge from some node in $A$ to some node in $B$. \lemmref{anyDelta} says that for any constant $\delta > 0$ we can construct a  family of DAGs $\{G_n^\delta\}_{n=1}^\infty$ with $\indeg=O(\log n)$ such that each $G_n^\delta$ is a $\delta$-local expander. \lemmref{anyDelta} essentially restates~\cite[Claim 1]{EGS75} except that we require that $G_n$ is a $\delta$-local expander for {\em all} $n >0$ instead of for $n$ sufficiently large. Since we require a (very) minor modification to achieve $\delta$-local expansion for {\em all} $n >0$ we sketch the proof of \lemmref{anyDelta} in \appref{MissingProof} for completeness. 

\newcommand{\LemmaDeltaExpander}{Let $\delta > 0$ be a fixed constant then there is a family of DAGs $\{G_n^\delta\}_{n=1}^\infty$ with $\indeg=O(\log n)$ such that each $G_n^\delta$ is a $\delta$-local expander.}
\begin{lemma} \cite{EGS75}\lemmlab{anyDelta}
\LemmaDeltaExpander
\end{lemma}

While \lemmref{anyDelta} essentially restates~\cite[Claim 1]{EGS75}, \lemmref{GoodNodesConnected} and \lemmref{CountGoodNodes} improve upon the analysis of~\cite{EGS75}. We say that a node $x \in [n]$ is $\gamma$-good under a subset $S\subseteq [n]$ if for all $r > 0$ we have $\left| I_r(x)\backslash S \right| \geq \gamma \left| I_r(x)\right|$ and  $\left| I_r^*(x)\backslash S \right| \geq  \gamma \left| I_r^*(x)\right|$. \lemmref{GoodNodesConnected} is similar to \cite[Claim 3]{EGS75}, which  also states that all $\gamma$-good nodes are connected by a directed path in $G-S$. However, we stress that the argument of \cite[Claim 3]{EGS75} requires that $\gamma \geq 0.5$ while \lemmref{GoodNodesConnected} has no such restriction. This is crucial to prove \thmref{ExtremeDepthRobust} where we will select $\gamma$ to be very small. 

\begin{lemma} \lemmlab{GoodNodesConnected} 
Let $G = (V=[n],E)$ be a $\delta$-local expander and let $x < y \in [n]$ both be  $\gamma$-good under $S \subseteq [n]$ then if $\delta < \min\{\gamma/2,1/4\}$ then there is a directed path from node $x$ to node $y$ in $G-S$. 
\end{lemma}

\lemmref{CountGoodNodes} shows that {\em almost all} of the nodes in $G-S$ are $\gamma$-good. It immediately follows that $G_n-S$ contains a directed path running through {\em almost all} of the nodes $[n]\setminus S$. While  \lemmref{CountGoodNodes} may appear similar to \cite[Claim 2]{EGS75} at first glance, we again stress one crucial difference. The proof of \cite[Claim 2]{EGS75} is only sufficient to show that at least $n- 2|S|/(1-\gamma) \geq n-2|S|$ nodes are $\gamma$-good. At best this would allow us to conclude that $G_n$ is $(e,n-2e)$-depth robust. Together \lemmref{CountGoodNodes} and \lemmref{GoodNodesConnected} imply that if $G_n$ is a $\delta$-local expander ($\delta < \min\{\gamma/2,1/4\}$) then $G_n$ is $\left(e,n-e\frac{1+\gamma}{1-\gamma}\right)$-depth robust.

\begin{lemma} \lemmlab{CountGoodNodes} For any DAG $G = ([n],E)$ and any subset $S \subseteq [n]$ of nodes at least $n-|S|\frac{1+\gamma}{1-\gamma}$ of the remaining nodes in $G$ are $\gamma$-good with respect to $S$.  
\end{lemma}

\begin{proofof}{\thmref{ExtremeDepthRobust}}
By \lemmref{anyDelta}, for any $\delta > 0$, there is a family of DAGs $\{J_n^\delta\}_{n=1}^\infty$ with $\indeg=O(\log n)$ such that for each $n \geq 1$ the DAG $J_n^\delta$ is a $\delta$-local expander. Given $\epsilon \in (0,1]$ we will set $G_n^{\epsilon} = J_n^\delta$ with $\delta=\epsilon/10 < 1/4$ so that $G_n^\epsilon$ is a $(\epsilon/10)$-local expander. We also set $\gamma = \epsilon/4 > 2\delta$. Let $S \subseteq V_n$ of size $|S| \leq e$ be given. Then by \lemmref{CountGoodNodes}  at least $n-e\frac{1+\gamma}{1-\gamma}$ of the nodes are $\gamma$-good and by\lemmref{GoodNodesConnected} there is a path connecting all $\gamma$-good nodes in $G-S$. Thus, the DAG $G_n^\epsilon$ is $\left(e,n-e\frac{1+\gamma}{1-\gamma}\right)$-depth robust for any $e \leq n$. In particular, if $\alpha = e/n$ and $\beta = 1-\alpha\frac{1+\gamma}{1-\gamma}$ then the graph is $(\alpha n,\beta n)$-depth robust. Finally we verify that
\[ n- \alpha n -\beta n = -e + e \alpha\frac{1+\gamma}{1-\gamma} = e\frac{2\gamma}{1-\gamma} \leq  n \frac{\epsilon}{2-\epsilon/2} \leq \epsilon n \ . \]
 \QED
\end{proofof}

The proof of \lemmref{GoodNodesConnected} follows by induction on the distance $|y-x|$ between $\gamma$-good nodes $x$ and $y$. Our proof extends a similar argument from \cite{EGS75} with one important difference. \cite{EGS75} argued inductively that for each good node $x$ and for each $r>0$ over half of the nodes in $I_r^*(x)$ are reachable from $x$ and that $x$ can be reached from over half of the nodes in $I_r(x)$ --- this implies that $y$ is reachable from $x$ since there is at least one node $z \in I_{|y-x|}^*(x)=I_{|y-x|}(y)$ such that $z$ can be reached from $x$ and $y$ can be reached from $z$ in $G-S$. Unfortunately, this argument inherently requires that $\gamma \geq 0.5$ since otherwise we may have at least $\left| I_r^*(x) \cap S \right| \geq (1-\gamma) r$ nodes in the interval $I_r(x)$ that are not reachable from $x$. To get around this limitation we instead show, see \claimref{numReachable}, that more than half of the nodes in the set $I_r^*(x) \setminus S$  are reachable from $x$ and that more than half of the nodes in the set  $I_r(x) \setminus S$ are reachable from $x$ --- this still suffices to show that $x$ and $y$ are connected since by the pigeonhole principle there is at least one node $z \in  I_{|y-x|}^*(x)\setminus S =I_{|y-x|}(y)\setminus S$ such that $z$ can be reached from $x$ and $y$ can be reached from $z$ in $G-S$.

\begin{claim} \claimlab{numReachable}
Let $G = (V=[n],E)$ be a $\delta$-local expander, let $x \in [n]$ be a $\gamma$-good node under $S \subseteq [n]$ and let $r >0$ be given. If $\delta < \gamma/2$ then all but $2\delta r$  of the nodes in $I_r^*(x)\backslash S$ are reachable from $x$ in $G-S$. Similarly, $x$ can be reached from all but $2\delta r$ of the nodes in $I_r(x)\backslash S$. In particular, if $\delta < 1/4$ then more than half of the nodes in $I_r^*(x)\backslash S$ (resp. in $I_r(x)\backslash S$) are reachable from $x$ (resp. $x$ is reachable from) in $G-S$.  
\end{claim}
\begin{proof}
We prove by induction that (1) if $r = 2^k \delta^{-1}$ for some integer $k$ then all but $\delta r$ of the nodes in $I_r^*(x)\backslash S$ are reachable from $x$ and, (2) if $2^{k-1}< r < 2^k \delta^{-1}$ then then all but $2\delta r$ of the nodes in $I_r^*(x)\backslash S$ are reachable from $x$. For the base cases we observe that if $r \leq \delta^{-1}$ then, by definition of a $\delta$-local expander, $x$ is directly connected to all  nodes in $I_r^*(x)$ so all nodes in $I_r(x)\backslash S $ are reachable. 

 Now suppose that claims (1) and (2) holds for each $r' \leq r= 2^k \delta^{-1}$. Then we show that the claim holds for each $r < r' \leq 2r=2^{k+1} \delta^{-1}$. In particular, let $A \subseteq I_r^*(x)\backslash S$ denote the set of nodes in $I_r^*(x)\backslash S$ that are reachable from $x$ via a directed path in $G-S$ and let $B \subseteq I_{r'-r}^*(x+r) \backslash S$ be the set of all nodes in $I_{r'-r}^*(x+r) \backslash S$ that are {\em not reachable} from $x$ in $G-S$. Clearly, there are no directed edges from $A$ to $B$ in $G$ and by induction we have $|A| \geq \left| I_{r}^*(x)\backslash S \right|-\delta r \geq r(\gamma-\delta) > \delta r$. Thus, by $\delta$-local expansion $|B| \leq r\delta$. Since, $\left| I_{r}^*(x)\backslash (S\cup A) \right|\leq \delta r$ at most $\left| I_{r'}^*(x)\backslash (S\cup A) \right| \leq |B|+\delta r \leq 2\delta r \leq 2 \delta r'$ nodes in $I_{2r}^*(x) \backslash S$ are not reachable from $x$ in $G-S$. Since, $r' > r$ the number of unreachable nodes is at most $2\delta r \leq 2\delta r'$, and if $r'=2r$ then the number of unreachable nodes is at most $2\delta r = \delta r'$. 
 
 A similar argument shows that $x$ can be reached from all but $2 \delta r$ of the nodes in $I_r(x)\backslash S$ in the graph $G-S$.
\end{proof}

\begin{proofof}{\lemmref{GoodNodesConnected}}
By \claimref{numReachable} for each $r$ we can reach $\left| I_r^*(x)\backslash S \right| - \delta  r =\left| I_r^*(x)\backslash S \right|\left(1- \delta \frac{\left| I_r^*(x)\right|}{\left| I_r^*(x)\backslash S \right|}\right) \geq \left| I_r^*(x)\backslash S \right| \left(1-\frac{\delta}{\gamma}\right) > \frac{1}{2}  \left| I_r^*(x)\backslash S \right| $ of the nodes in $I_r^*(x)\backslash S $ from the node $x$ in $G-S$.  Similarly, we can reach $y$ from more than $\frac{1}{2}  \left| I_r(x)\backslash S \right| $ of the nodes in $I_r(y)\backslash S$. Thus, by the pigeonhole principle we can find at least one node $z \in  I_{|y-x|}^*(x)\setminus S =I_{|y-x|}(y)\setminus S$ such that $z$ can be reached from $x$ and $y$ can be reached from $z$ in $G-S$. \QED
\end{proofof}
 
\lemmref{CountGoodNodes} shows that almost all of the nodes in $G-S$ are $\gamma$-good. The proof is again similar in spirit to an argument of \cite{EGS75}. In particular, \cite{EGS75} constructed a superset $T$ of the set of all $\gamma$-bad nodes and then bound the size of this superset $T$. However, they only prove that $BAD \subset T \subseteq F \cup B$ where $|F|,|B| \leq |S|/(1-\gamma)$. Thus, we have  $|BAD| \leq |T| \leq 2|S|/(1-\gamma)$. Unfortunately, this bound is not sufficient for our purposes. In particular, if $|S| = n/2$ then this bound does not rule out the possibility that $|BAD| = n$ so that none of the remaining nodes are good. Instead of bounding the size of the superset $T$ directly we instead bound the size of the set $T\setminus S$ observing that $|BAD| \leq |T| \leq |S| + |T\setminus S|$. In particular, we can show that $|T\setminus S| \leq \frac{2 \gamma|S|}{1-\gamma}$. We then have $|GOOD| \geq n-|T| = n-|S|-|T\backslash S| \geq n-|S|-\frac{2\gamma|S|}{1-\gamma}$. 

\begin{proofof}{\lemmref{CountGoodNodes}}
We say that a $\gamma$-bad node $x$ has a forward (resp. backwards) witness $r$ if $\left|I_r^*(x) \backslash S \right| > \gamma r$. Let $x_1^*,r_1^*$ be the lexicographically first $\gamma$-bad node with a forward witness. Once $x_1^*,r_1^*,\ldots,x_k^*,r_k^*$ have been define let $x_{k+1}^*$ be the lexicographically least $\gamma$-bad node such that $x_{k+1}^* > x_k^*+r_k^*$ and $x_{k+1}^*$ has  a forward witness $r_{k+1}^*$ (if such a node exists). Let    $x_1^*,r_1^*,\ldots,x_k^*,r_{k*}^*$ denote the complete sequence, and similarly define a maximal sequence $x_1,r_1,\ldots,x_k,r_{k}$ of $\gamma$-bad nodes with backwards witnesses such that $x_i-r_i > x_{i+1}$ for each $i$. 

Let \[F = \bigcup_{i=1}^{k^*} I_{r_i^*}^*\left(x_i^*\right) ~~~~\mbox{, and}~~~~B = \bigcup_{i=1}^{k} I_{r_i}\left(x_i\right) \] 
Note that for each $i \leq k^*$ we have $\left| I_{r_i^*}^*\left(x_i^*\right) \backslash S\right| \leq \gamma r$. Similarly, for each $i \leq k$ we have $\left| I_{r_i}\left(x_i\right) \backslash S\right| \leq \gamma r$. Because the sets $I_{r_i^*}^*\left(x_i^*\right) $ are all disjoint (by construction) we have \[\left| F \backslash S \right| \leq \gamma \sum_{i=1}^{k^*} r_i^* = \gamma|F| \ .\] Similarly,  $\left| B \backslash S \right| \leq \gamma |B|$. We also note that at least $(1-\gamma)|F|$ of the nodes in $|F|$ are in $|S|$. Thus, $|F|(1-\gamma) \leq |S|$ and similarly $|B|(1-\gamma) \leq |S|$. We conclude that $\left| F \backslash S \right| \leq \frac{\gamma |S|}{1-\gamma}$  and that $\left| B \backslash S \right| \leq \frac{\gamma |S|}{1-\gamma}$. 

To finish the proof let $T = F \cup B =S \cup \left(F \backslash S \right) \cup \left(B \backslash S \right)$. Clearly, $T$ is a superset of all $\gamma$-bad nodes. Thus, at least $n-|T| \geq n-|S|\left( 1 + \frac{2\gamma}{1-\gamma}\right) = n-|S|\frac{1+\gamma}{1-\gamma}$ nodes are good. 

\end{proofof}

We also remark that \lemmref{anyDelta} can be modified to yield a family of DAGs $\{G_n\}_{n=1}^\infty$ with $\indeg(G_n)=O\left( \log n \log^* n\right)$ such that $G_n$ is a $\delta_n$ local expander for some sequence $\{\delta_n\}_{n=1}^\infty$ converging to $0$. We can define a sequence $\{\gamma_n\}_{n=1}^\infty$ such that $\frac{1+\gamma_n}{1-\gamma_n}$ converges to $1$ and $2\gamma_n > \delta_n$ for each $n$. \lemmref{anyDelta} and \lemmref{CountGoodNodes} then imply that each $G_n$ is $\left(e,n-e\frac{1+\gamma_n}{1-\gamma_n}\right)$-depth robust for any $e \leq n$. 

\subsection{Additional Applications of Extremely Depth Robust Graphs}
We now discuss additional applications of \thmref{ExtremeDepthRobust}. 

\subsubsection{Application 0: Proofs of Sequential Work} As we previously noted Mahmoody et al.~\cite{MMV13} used extremely depth-robust graphs to construct efficient Proofs-Of-Sequential Work. In a proof of sequential work a prover wants to convince a verifier that he computed a hash chain of length $n$ involving the input value $x$ without requiring the verifier to recompute the entire hash chain. Mahmoody et al.~\cite{MMV13} accomplish this by requiring the prover computes labels $L_1,\ldots,L_n$ by ``pebbling'' an extremely depth-robust DAG $G_n$ e.g., $L_{i+1} = H\left(x \| L_{v_1} \| \ldots \| L_{v_\d}\right)$  where $\{v_1,\ldots,v_\d\} = \parents(i+1)$ and $H$ is a random oracle. The prover then commits to the labels $L_1,\ldots,L_n$ using a Merkle Tree and sends the root of the tree to the verifier who can audit randomly chosen labels e.g., the verifier audits label $L_{i+1}$ by asking the prover to reveal the values $L_{i+1}$ and $L_v$ for each $v \in \parents(i+1)$. If the DAG is extremely-depth robust then either a (possibly cheating) prover make at least $(1-\epsilon)n$ sequential queries to the random oracle, or the the prover will fail to convince the verifier with high probability~\cite{MMV13}. 

We note that the parameter $\d = \indeg( G_n)$ is crucial to the efficiency of the Proofs-Of-Sequential Work protocol since each audit challenge requires the prover to reveal $\d+1$ labels in the Merkle tree. The DAG $G_n$ from ~\cite{MMV13} has $\indeg( G_n) = O\left( \log^2 n \cdot \polylog \left(\log n\right)\right)$ while our DAG $G_n$ from \thmref{ExtremeDepthRobust} has maximum indegree $\indeg( G_n) = O\left( \log n\right)$. Thus, we can improve the communication complexity of the Proofs-Of-Sequential Work protocol  by a factor of $\Omega(\log n \cdot \polylog \log n)$.

\subsubsection{Application 1: Graphs with Maximum Cumulative Cost}
We now show that our family of extreme depth-robust DAGs has the highest possible cumulative pebbling cost even in terms of the {\em constant} factors. In particular, for any constant $\eta >0$ the family $\{G_n^\eta\}_{n=1}^\infty$ of DAGs from \thmref{ExtremeDepthRobust}  has $\ppeb_\Ccc(G_n) \geq \frac{n^2(1-\eta)}{2}$ and $\indeg(G_n)=O(\log n)$. By comparison,  $\ppeb_\Ccc(G_n) \leq \frac{n^2+n}{2}$ for {\em any} DAG $G \in \G_n$ --- even if $G$ is the complete DAG. 

Previously, Alwen et al.~\cite{ABP17} showed that any $(e,d)$-depth robust DAG $G$ has $\ppeb_\Ccc(G) > ed$ which implies that their is a family of  DAG $G_n$ with $\ppeb_\Ccc(G_n) = \Omega\left( n^2 \right)$~\cite{EGS75}. We stress that we need new techniques to prove \thmref{optimalCC}. Even if a DAG $G \in \G_n$ were $(e,n-e)$-depth robust for every $e \geq 0$ (the only DAG actually satisfying this property is the compete DAG $K_n$)~\cite{ABP17} only implies that $\ppeb_\Ccc(G_n) \geq \max_{e \geq 0} e(n-e) = n^2/4$. Our basic insight is that at time $t_i$, the first time a pebble is placed on node $i$ in $G_n^\epsilon$, the node $i+\gamma i$ is $\gamma$-good and is therefore reachable via an undirected path from all of the other $\gamma$-good nodes in $[i]$. If we have $|P_{t_i}| < \left( 1-\eta/2\right) i$ then we can show that at least $\Omega(\eta i)$ of the nodes in $[i]$ are $\gamma$-good. We can also show that these $\gamma$-good nodes form a depth robust subset and will cost $\Omega\left((\eta-\epsilon)^2i^2\right)$ to repebble them by~\cite{ABP17}. Since, we would need to pay this cost by time $t_{i+\gamma i}$ it is less expensive to simply ensure that $|P_{t_i}| >  \left( 1-\eta/2\right) i$. We refer an interested reader to \appref{MissingProof} for a complete proof.
  
\newcommand{\thmrefOptimalCC}{For any constant $0 < \eta < 1$ the family $\{G_n^\eta\}_{n=1}^\infty$ of DAGs from \thmref{ExtremeDepthRobust} has $\indeg(G_n) = O\left( \log n\right)$ and $\ppeb_\Ccc(G_n) \geq \frac{n^2\left(1-\eta\right)}{2}$.}
\begin{theorem} \thmlab{optimalCC}
\thmrefOptimalCC
\end{theorem}

\subsubsection{Cumulative Space in Parallel-Black Sequential-White Pebblings}
The black-white pebble game~\cite{cook1976storage} was introduced to model nondeterministic computations. White pebbles correspond to nondeterministic guesses and can be placed on any vertex at any time, but these pebble can only be removed when (e.g., when we can verify the correctness of this guess). Formally, black white-pebbling configuration $P_i = \left( P_i^W,P_i^B\right)$ of a DAG $G=([n],E)$ consists of two subsets $P_i^W, P_i^B \subseteq [n]$ where $P_i^B$ (resp. $P_i^W$) denotes the set of nodes in $G$ with black (resp. white) pebbles on them at time $i$. For a legal  parallel-black sequential-white pebbling $P = (P_0,\ldots,P_t) \in  \Peb_{G}^{BW}$ we require that  we start with no pebbles on the graph i.e., $P_0 = (\emptyset,\emptyset)$ and that all white pebbles are removed by the end i.e., $P_t^W = \emptyset$ so that we verify the correctness of every nondeterministic guess before terminating. If we place a black pebble on a node $v$ during round $i+1$ then we require that all of $v$'s parents have a pebble (either black or white) on them during round $i$ i.e., $\parents\left(P_{i+1}^B \setminus P_{i}^B\right) \subseteq P_{i}^B \cup P_i^W$. In the  Parallel-Black Sequential-White model we require that at most one new white pebble is placed on the DAG in every round i.e., $\left|P_i^W \setminus P_{i-1}^W \right| \leq 1$ while no such restrict applies for black pebbles.  See \defref{pebbling} in \appref{MissingProof} for a more formal definition of the parallel-black sequential white pebbling game. 

We can use our construction of a family of extremely depth-robust DAG $\{G_n\}_{n=1}^\infty$ to establish new upper and lower bounds for 
 
Alwen et al.~\cite{alwen2017cumulative} previously showed that in the parallel-black sequential white pebbling model an $(e,d)$-depth-robust DAG $G$ requires cumulative space at least $\peb^{BW}_{cc}(G) \doteq \min_{P \in \Peb_{G}^{BW}}\sum_{i=1}^t \left| P_i^B \cup P_i^W \right| =  \Omega\left(e\sqrt{d}\right)$ or at least $\geq ed$ in the sequential black-white pebbling game. In this section we show that any $(e,d)$-reducible DAG admits a parallel-black sequential white pebbling with cumulative space at most $O(e^2+dn)$ which implies that any DAG with constant indegree admits a parallel-black sequential white pebbling with cumulative space at most $O(\frac{n^2 \log^2 \log n}{\log^2 n})$ since any DAG is $(n \log \log n/\log n, n/\log^2 n)$-reducible. We also show that this bound is essentially tight (up to $\log \log n$ factors) using our construction of extremely depth-robust DAGs. In particular, we can find a family of DAGs $\{G_n\}_{n=1}^\infty$ with $\indeg(G_n)=2$ such that any  parallel-black sequential white pebbling has cumulative space at least $\Omega(\frac{n^2}{\log^2 n})$. To show this we start by showing that any  parallel-black sequential white pebbling of an extremely depth-robust DAG $G$, with $\indeg(G)=O(\log n)$, has cumulative space at least $\Omega(n^2)$. We use \lemmref{IndegRed} to reduce the indegree of the DAG and obtain a DAG $G'$ with $n'=O(n \log n)$ nodes and $\indeg(G)=2$, such that any parallel-black sequential white pebbling of $G'$ has cumulative space at least $\Omega(\frac{n^2}{\log^2 n})$.

To the best of our knowledge no general upper bound on cumulative space complexity for parallel-black sequential-white pebblings was known prior to our work other than the parallel black-pebbling attacks of Alwen and Blocki~\cite{AB16}. This attack,  which doesn't even use the white pebbles, yields an upper bound of $O(ne+n\sqrt{nd})$ for $(e,d)$-reducible DAGs and $O(n^2 \log\log n/\log n)$ in general. One could also consider a ``parallel-white parallel-black'' pebbling model in which we are allowed to place as many white pebbles as he would like in each round. However, this model admits a trivial pebbling. In particular, we could place white pebbles on every node during the first round and remove all of these pebbles in the next round e.g., $P_1 = (\emptyset,V)$ and $P_2 = (\emptyset,\emptyset)$. Thus, any DAG has cumulative space complexity $\theta(n)$ in the ``parallel-white parallel-black'' pebbling model.

\thmref{whitePebbleUpper} shows that $(e,d)$-reducible DAG admits a parallel-black sequential white pebbling with cumulative space at most $O(e^2+dn)$. The basic pebbling strategy is reminiscent of the parallel black-pebbling attacks of Alwen and Blocki~\cite{AB16}. Given an appropriate depth-reducing set $S$ we use the first $e=|S|$ steps to place white pebbles on all nodes in $S$. Since $G-S$ has depth at most $d$ we can place black pebbles on the remaining nodes during the next $d$ steps. Finally, once we place pebbles on every node we can legally remove the white pebbles. A formal proof of \thmref{whitePebbleUpper} can be found in \appref{MissingProof}.
\newcommand{\thmWhitePebbleUpper}{Let $G= (V,E)$ be $(e,d)$-reducible then $\peb^{BW}_{cc}(G) \leq \frac{e(e+1)}{2} + dn$. In particular, for any DAG $G$ with $\indeg(G)=O(1)$ we have $\peb^{BW}_{cc}(G) = O\left( \left(\frac{n \log\log n}{\log n} \right)^2 \right)$.}
\begin{theorem} \thmlab{whitePebbleUpper}
\thmWhitePebbleUpper
\end{theorem}

\thmref{HighCCBW} shows that our upper bound is essentially tight. In a nut-shell their lower bound was based on the observation that for any integers $i,d$ the DAG $G-\bigcup_j P_{i+jd}$ has depth at most $d$ since any remaining path must have been pebbled completely in time $d$--- if $G$ is $(e,d)$-depth robust this implies that $\left|\bigcup_j P_{i+jd}\right| \geq e$. The key difficulty in adapting this argument to the parallel-black sequential white pebbling model is that it is actually possible to pebble a path of length $d$ in $O(\sqrt{d})$ steps by placing white pebbles on every interval of length $\sqrt{d}$. This is precisely why Alwen et al.~\cite{alwen2017cumulative} were only able to establish the lower bound $\Omega(e\sqrt{d})$ for the cumulative space complexity of $(e,d)$-depth robust DAGs ---  observe that we always have $e\sqrt{d} \leq n^{1.5}$ since $e+d \leq n$ for any DAG $G$. We overcome this key challenge by using extremely depth-robust DAGs. 

In particular, we exploit the fact that extremely depth-robust DAGs are ``recursively'' depth-robust. For example, if $G$ is $(e,d)$-depth robust for any $e+d \leq (1-\epsilon)n$ then the DAG $G-S$ is $(e,d)$-depth robust for any $e+d \leq (n-|S|) - \epsilon n$. Since $G-S$ is still sufficiently depth-robust we can then show that for some node $x \in V(G-S)$ any (possibly incomplete) pebbling $P= (P_0,P_1, \ldots, P_t)$ of $G-S$ with $P_0 = P_t = (\emptyset,\emptyset)$ either (1) requires $t=\Omega(n)$ steps, or (2) fails to place a pebble on $x$ i.e. $x \notin \bigcup_{r=0}^t \left(P_0^W \cup P_r^B\right)$. By \thmref{ExtremeDepthRobust} it then follows that there is a family of DAGs $\{G_n\}_{n=1}^\infty$ with $\indeg(G_n) = O\left(\log n\right)$ and $\peb^{BW}_{cc}(G)  = \Omega(n^2)$. If apply indegree reduction \lemmref{IndegRed} to $G_n$ we obain a DAG $G'_n$ with $\indeg(G'_n)=2$ and $O(n \log n)$ nodes. A similar argument shows that $\peb^{BW}_{cc}(G)  = \Omega(n^2/\log^2 n)$. A formal proof of \thmref{HighCCBW} can be found in \appref{MissingProof}.

\newcommand{\ThmHighCCBW}{Let $G= (V = [n],E \supset \{(i,i+1): i < n\})$ be $(e,d)$-depth-robust for any $e+d \leq (1-\epsilon)n$ then $\peb^{BW}_{cc}(G) \geq  \left(1/16-\epsilon/2  \right)n^2 $. Furthermore, if $G'= ([2n\d],E')$ is the indegree reduced version of $G$ from \lemmref{IndegRed} then $\peb^{BW}_{cc}(G') \geq  \left(1/16-\epsilon/2  \right)n^2$. In particular, there is a family of DAGs $\{G_n\}_{n=1}^\infty$ with $\indeg(G_n) = O\left(\log n\right)$ and $\peb^{BW}_{cc}(G)  = \Omega(n^2)$, and a separate family of DAGs $\{H_n\}_{n=1}^\infty$ with $\indeg(H_n) = 2$ and $\peb^{BW}_{cc}(H_n)  = \Omega\left(\frac{n^2}{\log^2 n}\right)$.}
\begin{theorem} \thmlab{HighCCBW}
\ThmHighCCBW
\end{theorem}

%% file: ss-pebbling.tex
%
As an application of the pebbling results on sustained space in this section we construct a new type of moderately hard function (MoHF) in the parallel random oracle model pROM. In slightly more detail, we first fix the computational model and define a particular notion of moderatly hard function called sustained memory-hard functions (SMHF). We do this using the framework of~\cite{AT17} so, beyond the applications to password based cryptography, the results in~\cite{AT17} for building provably secure cryptographic applications on top of any MoHF can be immediatly applied to SMHFs. In particular this results in a proof-of-work and non-interactive proof-of-work where ``work'' intuitively means having performed some computation entailing sufficient sustained memory. Finally we prove a ``pebbling reduction'' for SMHFs; that is we show how to bound the parameters describing the sustained memory complexity of a family of SMHFs in terms of the sustained space of their underlying graphs.\footnote{Effectively this does for SMHFs what~\cite{AT17} did for MHFs.}

\subsection{Defining Sustained Memory Hard Functions}
We very briefly sketch the most important parts of the MoHF framework of~\cite{AT17} which is, in turn, a generalization of the indifferentiability framework of~\cite{MaReHo04}.

We begin with the following definition which describes a family of functions that depend on a (random) oracle.

\begin{definition}[Oracle functions]\deflab{oraclefunc}
 For (implicit) oracle set $\Hc$, an \emph{oracle function $f^{(\cdot)}$ (with domain $D$ and range $R$)}, denoted $f^{(\cdot)}:D \to R$, is a set of functions indexed by oracles $h\in\Hc$ where each $f^h$ maps $D \to R$.
\end{definition}

Put simply, an MoHF is a pair consisting of an oracle family $f^{(\cdot)}$ and an honest algorithm $\naive$ for evaluating functions in the family using access to a random oracle. Such a pair is secure relative to some computational model $M$ if no adversary $\A$ with a computational device adhering to $M$ (denoted $\A\in M$) can produce output which couldn't be produced simply by called $f^{(h)}$ a limited number of times (where $h$ is a uniform choice of oracle from $\Hc$). It is asumed that algorithm $\naive$ is computable by devices in some (possibly different) computational model $\bar{M}$ when given sufficent computational resources. Usually $M$ is strictly more powerful than $\bar{M}$ reflecting the assumption that an adversary could have a more powerful class of device than the honest party. For example, in this work we will let model $\bar{M}$ contain only sequential devices (say Turing machines which make one call to the random oracle at a time) while $M$ will also include parallel devices. 

In this work, both the computational models $M$ and $\bar{M}$ are parametrized by the same space $\Pc$. For each model, the choice of parameters fixes upperbounds on the power of devices captured by that model; that is on the computational resources available to the permitted devices. For example $M_{a}$ could be all Turing machines making at most $a$ queries to the random oracle. The security of a given moderatly hard function is parameterized by two functions $\a$ and $\b$ mapping the parameter space for $M$ to positive integers. Intuitively these functions are used to provide the following two properties.
\begin{description}
 \item[{\sc Completeness:}] To ensure the construction is even useable we require that $\naive$ is (computable by a device) in model $M_{a}$ and that $\naive$ can evaluate $f^{(h)}$ (when given access to $h$) on at least $\a(a)$ distinct inputs.
 \item[{\sc Security:}] To capture how bounds on the resources of an adversary $\A$ limit the ability of $\A$ to evalute the MoHF we require that the output of $\A$ when running on a device in model $M_{b}$ (and having access to the random oracle) can be reproduced by some simulator $\simul$ using at most $\b(b)$ oracle calls to $f^{(h)}$ (for uniform randomly sampled $h\from \Hc$. 
\end{description}
\noindent To help build provably secure applications on top of MoHFs the framework makes use of a destinguisher $\dist$ (similar to the environment in the Universal Composability\cite{Can01} family of models or, more accurately, to the destinguisher in the indifferentiability framework). The job of $\dist$ is to (try to) tell a \emph{real world} interaction with $\naive$ and the adversary $\A$ apart from an \emph{ideal world} interaction with $f^{(h)}$ (in place of $\naive$) and a simulator (in place of the adversary). Intuitivelly, $\dist$'s access to $\naive$ captures whatever $\dist$ could hope to learn by interacting with an arbitrary application making use of the MoHF. The definition then ensures that even leveraging such information the adversary $\A$ can not produce anything that could not be simulated (by simulator $\simul$) to $\dist$ using nothing more than a few calls to $f^{(h)}$. 

As in the above description we have ommited several details of the framework we will also use a somewhat simplified notation. We denote the above described real world execution with the pair $(\naive, \A)$ and an ideal world execution where $\dist$ is permited $c\in\N$ calls to $f^{(\cdot)}$ and simulator $\simul$ is permited $d\in\N$ calls to $f^{(h)}$ with the pair $(f^{(\cdot)}, \simul)_{c,d}$. To denote the statement that no $\dist$ can tell an interaction with $(\naive, \A)$ apart one with $(f^{(\cdot)}, \simul)_{c,d}$ with more than probability $\e$ we write $(\naive, \A) \indist[\e] (f^{(\cdot)}, \simul)_{c,d}$.

Finally, to accomadate honest parties with varying amounts of resources we equip the MoHF with a hardness parameter $n\in\N$. The following is the formal security definition of a MoHF. Particular types of MoHF (such as the one we define bellow for sustained memory complexity) differ in the precise notion of computational model they consider. For further intution, a much more detailed exposition of the framework and how the following definition can be used to prove security for applications we refer to~\cite{AT17}.

\begin{definition}[MoHF security]\deflab{security-parametrized}\deflab{security}
  Let $M$ and $\bar{M}$ be computational models with bounded resources parametrized by $\Pc$. For each $n \in \N$, let $f_n^{(\cdot)}$ be an oracle function and $\naive(n,\cdot)$ be an algorithm (computable by some device in $\bar{M}$) for evaluating $f_n^{(\cdot)}$. Let $\a,\b: \Pc \times \N \to \N$, and let $\eps: \Pc \times \Pc \times \N \to \R_{\geq 0}$. Then, $(f^{(\cdot)}_n, \naive_n)_{n\in\N}$ is a \emph{$(\a,\b,\eps)$-secure moderately hard function family (for model $M$)} if
    \begin{equation}
        \forall n\in\N, \rv \in \Pc, \A \in M_{\rv}\ \exists \simul \ \forall \lv \in \Pc: \quad
        (\naive(n,\cdot),\A) ~\indist[\eps(\lv,\rv,n)]~(f_n^{(\cdot)},\simul)_{\a(\lv,n),\b(\rv,n)} \; ,
    \end{equation}
    The function family is \emph{asymptotically secure} if $\eps(\lv,\rv,\cdot)$ is a negligible function in the third parameter for all values of $\rv,\lv\in\Pc$.
\end{definition}

\paragraph{Sustained Space Constrained Computation.} Next we define the honest and adversarial computational models for which we prove the pebbling reduction. In particular we first recall (a simplified version of) the pROM of~\cite{AT17}. Next we define a notion of sustained memory in that model naturally mirroring the notion of sustained space for pebbling. Thus we can parametrize the pROM by memory threshold $s$ and time $t$ to capture all devices in the pROM with no more sustained memory complexity then given by the choice of those parameters.

In more detail, we consider a resource-bounded computational device $\ares$. Let $w\in\N$. Upon startup, $\promres$ samples a fresh random oracle $h \getsr \Hc_w$ with range $\{0,1\}^w$. Now $\promres$ accepts as input a pROM algorithm $\A$ which is an oracle algorithm with the following behavior.

A {\em state} is a pair $(\t, \bfs)$ where {\em data} $\t$ is a string and $\bfs$ is a tuple of strings. The output of step $i$ of algorithm $\A$ is an {\em output state} $\bs_i = (\t_i, \bfq_i)$ where $\bfq_i = [q_i^1,\ldots,q_i^{z_i}]$ is a tuple of {\em queries} to $h$. As input to step $i+1$, algorithm $\A$ is given the corresponding {\em input state} $\s_i = (\t_i, h(\bfq_i))$, where $h(\bfq_i) = [h(q_i^1), \ldots, h(q_i^{z_i})]$ is the tuple of {\em responses} from $h$ to the queries $\bfq_i$. In particular, for a given $h$ and random coins of $\A$, the input state $\s_{i+1}$ is a function of the input state $\s_i$. The initial state $\s_0$ is empty and the input $x_{\inp}$ to the computation is given a special input in step $1$.

For a given execution of a pROM, we are interested in the following new complexity measure parametrized by an integer $s\ge 0$. We call an element of $\{0,1\}^s$ a \emph{block}. Moreover, we denote the bit-length of a string $r$ by $|r|$. The \emph{length} of a state $\s=(\t,\bfs)$ with $\bfs=(s^1, s^2, \ldots, s^y)$ is $|\s| = |\t| + \sum_{i\in[y]} |s^i|$. For a given state $\s$ let $b(\s) = \flr{|\s| / s}$ be the number of ``blocks in $\s$''. Intuitively, the $s$-sustained memory complexity ($s$-SMC) of an execution is the sum of the number of blocks in each state. More precisely, consider an execution of algorithm $\A$ on input $x_{\inp}$ using coins $\coins$ with oracle $h$ resulting in $z\in\Z_{\ge 0}$ input states $\s_1, \ldots, \s_z$, where $\s_i = (\t_i,\bfs_i)$ and $\bfs_i = (s_i^1, s_i^2, \ldots, s_i^{y_j})$. Then the for integer $s\ge 0$ the \emph{$s$-sustained memory complexity} ($s$-SMC) of the execution is
$$s\-\smc(\A^h(x_{\inp};\coins)) = \ds\sum_{i\in[z]} b(\s_i) \; ,$$
while the \emph{total number of RO calls} is $\sum_{i\in[z]} y_j$. More generally, the $s$-SMC (and total number of RO calls) of several executions is the sum of the $s$-sMC (and total RO calls) of the individual executions.

We can now describe the resource constraints imposed by $\promres$ on the pROM algorithms it executes. To quantify the constraints, $\promres$ is parametrized by element from $\prompar=\N^3$ which describe the limites on an execution of algorithm $\A$. In particular, for parameters $(q, s, t)\in\prompar$, algorithm $\A$ is allowed to make a total of $q$ RO calls and have $s$-SMC at most $t$ (summed across all invocations of $\A$ in any given experiment).

As usual for moderately hard functions, to ensure that the honest algorithm can be run on realistic devices, we restrict the honest algorithm $\naive$ for evaluating the SMHF to be a \emph{sequential} algorithms. That is, $\naive$ can make only a single call to $h$ per step. Technically, in any execution, for any step $j$ it must be that $y_j \leq 1$. No such restriction is placed on the adversarial algorithm reflecting the power (potentially) available to such a highly parallel device as an ASIC. In symbols we denote the sequential version of the pROM, which we refer to as the sequential ROM (sROM) by $\sromres$.

We can now (somewhat) formally define of a sustained memory-hard function for the pROM. The definition is a particular instance of and moderately hard function (c.f.~\defref{security}).

\begin{definition}[Sustained Memory-Hard Function]\deflab{mehf}
    For each $n \in \N$, let $f_n^{(\cdot)}$ be an oracle function and $\naive_n$ be an sROM algorithm for computing $f^{(\cdot)}$. Consider the function families: 
    $$\a = \{\a_w : \prompar \times \N \to \N\}_{w\in\N} \; ,~~~ \b = \{\b_w : \prompar \times \N \to \N\}_{w\in\N} \; ,$$ $$\e = \{\e_w:\prompar \times \prompar \times \N \to \N\}_{w\in\N} \; .$$
    Then $F=(f^{(\cdot)}_n, \naive_n)_{n\in\N}$ is called an \emph{$(\a,\b,\e)$-sustained memory-hard function} (SMHF) if $\forall w\in\N$ $F$ is an $(\a_w,\b_w,\e_w)$-secure moderately hard function family for $\promres$.
\end{definition}

\subsection{The Construction}
In this work $f^{(\cdot)}$ will be a graph function~\cite{AS15} (also sometimes called ``hash graph''). The following definition is taken from~\cite{AT17}. A graph function depends on an oracle $h\in\Hc_w$ mapping bit strings to bit strings. We also assume the existance of an implicit prefix-free encoding such that $h$ is evaluated on unique strings. Inputs to $h$ are given as distinct tuples of strings (or even tuples of tuples of strings). For example, we assume that $h(0,00)$, $h(00,0)$, and $h((0,0),0)$ all denote distinct inputs to $h$.
\begin{definition}[Graph function]\deflab{labeling}
 Let function $h:\{0,1\}^* \to \{0,1\}^w\in\Hc_w$ and DAG $G=(V,E)$ have source nodes $\{v^{\sfin}_1, \ldots, v^{\sfin}_a\}$ and sink nodes $(v^{\sfout}_1, \ldots, v^{\sfout}_z)$. Then, for inputs $\bfx = (x_1,\ldots,x_a) \in (\{0,1\}^*)^{\times a}$, the \emph{$(h, \bfx)$-labeling} of $G$ is a mapping $\lab : V \to \{0,1\}^w$ defined recursively to be:
 \begin{equation*}
   \forall v \in V ~~ \lab(v)  := 
    \begin{cases}
     h(\bfx, v, x_{j}))                    & : v = v^{\sfin}_j\\
     h(\bfx, v, \lab(v_1), \ldots, \lab(v_d)))          & : \mbox{else}
    \end{cases}
\end{equation*}
where $\{v_1, \ldots, v_d\}$ are the parents of $v$ arranged in lexicographic order.

The \emph{graph function (of $G$ and $\Hc_w$)} is the oracle function $$f_G:(\{0,1\}^*)^{\times a} \to (\{0,1\}^w)^{\times z} \; ,$$ which maps $\bfx \mapsto (\lab(v^{\sfout}_1), \ldots, \lab(v^{\sfout}_z))$ where $\lab$ is the $(h,\bfx)$-labeling of $G$.
\end{definition}

Given a graph function we need an honest (sequential) algorithm for computing it in the pROM. For this we use the same algorithm as already used in~\cite{AT17}. The honest oracle algorithm $\naive_G$ for graph function $f_G$ computes one label of $G$ at a time in topological order appending the result to its state. If $G$ has $|V|=n$ nodes then $\naive_G$ will terminate in $n$ steps making at most $1$ call to $h$ per step, for a total of $n$ calls, and will never store more than $n*w$ bits in the data portion of its state. In particular for all inputs $\bfx$, oracles $h$ (and coins $\coins$) we have that for any $s\in[n]$ if the range of $h$ is in $\{0,1\}^w$ then algorithm $\naive$ has $sw$-SMC of $n-s$.

Recall that we would like to set $\a_w:\prompar \to \N$ such that for any parameters $(q,s,t)$ constraining the honest algorithms resources we are still guaranteed at least $\a_w(q,s,t)$ evaluations of $f_G$ by $\naive_G$. Given the above honest algorithm we can thus set:
 \begin{equation*}
   \forall (q,s,t) \in \prompar ~~ \a_w(q,s,t)  := 
    \begin{cases}
     0                                                  & : q < n\\
     \min(\flr{q/n},\flr{t/(n-\flr{s/w}})               & : \mbox{else}
    \end{cases}
\end{equation*}

\noindent It remains to determine how to set $\b_w$ and $\e_w$, which is the focus of the remainder of this section.

\subsection{The Pebbling Reduction}
We state the main theorem of this section which relates the parameters of an SMHF based on a graph function to the sustained (pebbling) space complexity of the underlying graph.

\newcommand{\PebblingReductionThm}{[Pebbling reduction] Let $G_n=(V_n,E_n)$ be a DAG of size $|V_n|=n$. Let $F=(f_{G,n}, \naive_{G,n})_{n\in\N}$ be the graph functions for $G_n$ and their na\"ive oracle algorithms. Then, for any $\lambda \ge 0$, $F$ is an $(\a,\b,\e)$-sustained memory-hard function where
  $$\a  = \left\{\a_w(q,s,t)\right\}_{w\in\N} \; ,$$
  $$\b = \left\{\b_w(q,s,t) = \frac{\ppeb_\Css(G,s)(w-\log q)}{1+\lambda}\right\}_{w\in\N} \; ,~~~ \e = \left\{\e_w(q,m) \le \frac{q}{2^w} + 2^{-\lambda}\right\}_{w\in\N}\; .$$}
\begin{theorem}\thmlab{pebred}
\PebblingReductionThm
\end{theorem}

The technical core of the proof follows that of~\cite{AT17} closely. For completeness we briefly sketch the proof in~\appref{pebproof}.

%% file: ss-open.tex
%

We conclude with several open questions for future research. The primary challenge is to provide a practical construction of a DAG $G$ with high sustained space complexity. While we provide a DAG $G$ with asymptotically optimal sustained space complexity, we do not optimize for constant factors. We remark that for practical applications to iMHFs it should be trivial to evaluate the function $\parents_G(v)$ without storing the DAG $G$ in memory explicitly. Toward this end it would be useful to either prove or refute the conjecture that any depth-robustness is sufficient for high sustained space complexity e.g., what is the sustained space complexity of the depth-robust DAGs from \cite{EGS75} or \cite{PTC76}? Another interesting direction would be to relax the notion of sustained space complexity and instead require that for any pebbling $P \in \pPeb(G)$ either (1) $P$ has large cumulative complexity e.g., $n^3$, or (2) $P$ has high sustained space complexity. Is it possible to design a dMHF with the property for any evaluation algorithm either has (1)  sustained space complexity $\Omega(n)$ for $\Omega(n)$ rounds, or (2) has cumulative memory complexity $\omega(n^2)$?

%% file: ss-missing.tex
%
\begin{reminderlemma}{\lemmref{indegred}}[Indegree Reduction for Parallel Sustained Space]
\IndegReductionForParallelSustainedSpace
\end{reminderlemma}
\begin{proofof}{\lemmref{indegred}}
 To obtain $H$ from $G$ we replace each node $v$ in $G$ with a path of length $\indeg(v)$ and distribute the incoming edges of $v$ along the path. More precicely let $G=(V,E)$ with sinks $S\subseteq V$. For each $v\in V$ let $\d_v = \indeg(v)$ and $p_{v,i}\in V$ be the $i$\th parent of $v$ (sorted in some arbitrary fixed order). By convention $p_{s,0} = \bot$ for all $s\in S$. We define $H=(V',E')$ as follows. The set of nodes $V' \subseteq V\times [\d]\cup \{\bot\}$ is
 $$V' = \{ \node{s}{\bot} : s\in S\} \cup \{\node{v}{i} : v\in V\setminus S, i \in [\d_v]\big\}.$$
 The edge set is given by:
 $$
   E' = \big\{(\node{v}{i-1},\node{v}{i}) : v\in V\setminus S, i \in [\d_v] \big\} \bigcup \big\{(\node{u}{\d_u}, \node{v}{i}) : (u,v)\in E, u = p_{v,i}\big\}.
 $$
 Each node of $G$ is replaced by at most $\d$ nodes in $H$ so the size $n'$ of $H$ is $n'\in[n,\d n]$. Moreover, by construction, no node in $H$ has more than two incoming edges so $H\in\G_{n',2}$ as desired.
 
 Next we map any $P'\in\pPeb_H$ to a $P\in\pPeb_G$ and show that $\forall s\ge 0$ we have $\ppeb_\Css(P',s) \ge \ppeb_\Css(P,s/(\d-1))$. In more detail, given $P'=(P'_0,\ldots,P'_z)\in\pPeb_H$ we define $P=(P_0, \ldots, P_z)$ as follows. 
 \begin{enumerate}
  \item For all $i\in[0,z]$ if $\node{v}{\d_v} \in P'_i$ then put $v$ in $P_i$.
  \item Further if $\node{v}{j}\in P'_i$ for $j< \d_v$ then put $(u_1,u_2,\ldots,u_j)$ in to $P_i$.
 \end{enumerate}
 
 \begin{claim}
 $P'\in\pPeb_H ~~\implies~~ P\in\pPeb_G$. 
 \end{claim}
 \begin{proof}
  By assumption $P'_0 = \emptyset$ so $P_0=\emptyset$. Moreover when a sink $\node{v}{\d_v} \in V'$ of $H$ is pebbled by $P'$ at time $i$ then the sink $v\in V$ of $G$ is pebbled in $P'$. But any sink of $G$ is mapped to a path in $H$ ending in a sink of $H$. Thus if all sinks of $H$ are pebbled by $P'$ then so must all sinks of $G$ be pebbled by $P$. In particular, as by assumption $P'$ is complete so is $P$. 
  
  To prove the claim it remains to show that if $P'$ is a legal pebbling for $H$ then so is $P$ a legal pebbling of $G$. Suppose, for the sake of contradiction that this is not the case and let $i\in[0,z]$ be the first time a pebble is placed illegally by $P$ and let it be on node $v\in V$. Suppose it was placed due to rule 1. Then it must be that $\node{v}{\d_v}\in P'_i$. Further, as $v\not\in P_{i-1}$ it must also be that $\node{v}{\d_v}\not\in P'_{i-1}$. By assumption $P'$ is legal so $\parents_H(\node{v}{\d_v})$ must be pebbled in $P_{i-1}$. If $\d_v = \bot$ then $v$ is a source node which contradicts it being pebbled illegally. If $\d_v = 1$ then there exists node $u = p_{v,1}\in V$ and $\node{u}{\d_u}\in P'_{i-1}$ which, according to rule $1$ above implies that $u\in P_{i-1}$. However that too is a contradiction to $v$ being pebbled illegally. If $\d_v > 1$ then both $\node{u}{\d_v}$ and $\node{v}{\d_v-1}$ are in $P'_{i-1}$. But by rules $1$ and $2$ then all parents of $v$ are pebbled in $P_{i-1}$ which is again a contradiction to $v$ being pebbled illegally at time $i$. Thus no node can be illegally pebbled due to rule $1$. 
  
  Let us suppose instead that $v$ was pebbled illegally due to rule $2$ being applied to a pebbled $\node{u}{i}$. That is for some $j\in[i]$ we have $v=p_{u,j}$. Since $v\not\in _P{i-1}$ and $P'$ is legal it must be that $j=i$. Moreover, it must be that $\parents_H(\node{u}{j}) \in P'_{i-1}$. In particular, then $\node{v,\d_v}\in P'_{i-1}$. But then rule $1$ implies that $v\in P_{i-1}$ which contradicts $v$ being pebbled illegally by $P$ at time $i$.
 \end{proof}
 
 To complete the proof of the lemma it remains only to relate the threshold complexities of $P$ and $P'$. Notice that for all $i\in[0,z]$ and any $v\in P_i$ at most $\d-1$ new pebbles where added to $P_i$. Thus we have that $\forall s\ge 0$ it holds that $\ppeb_\Css(P',s/(\d-1)) \ge \ppeb_\Css(P,s)$. \QED
\end{proofof}

\begin{reminderlemma}{\lemmref{anyDelta}} \cite{EGS75}
\LemmaDeltaExpander
\end{reminderlemma}

\begin{proofof}{\lemmref{anyDelta}} (sketch)
We closely follow the construction/proof of \cite{EGS75}. In particular, we say that a bipartite DAG $T_m^\delta = (V= A \cup B, E)$ with $|A|=|B|=m$ is a $\delta$-expander if for all $X \subseteq A, Y \subseteq B$ such that $|X| \geq \delta m$ and $ |Y| \geq \delta m$ we have $E \cap X \times Y \neq \emptyset$ i.e., there is an edge from some node $x\in X$ to some node $y \in Y$. For any constant $\delta$ we can find constants $c_\delta, m_\delta$ such that for all $m > m_\delta$ there exists a  $T_m^\delta$ expander with $\indeg\left(T_m^\delta\right) \leq c_\delta$ e.g., see the first lemma in \cite{EGS75} \footnote{In fact, the argument is probabilistic and there is a randomized algorithm which, except with negligibly small probability $\negl(m)$, constructs a  $\delta$-expander $T_m^\delta$ with $m$ nodes and $\indeg\left(T_m^\delta\right)  \leq 2 c_\delta$. }. Now following \cite{EGS75} we construct $G_n^\delta = ([n],E_n)$ by repeating the following steps for each $j \in (\left\lfloor \log_2 m_\delta \right\rfloor, \left\lfloor \log_2 m_{\delta/10} \right\rfloor)$. 
\begin{enumerate} 
\item We partition the nodes $[n]$ into $r=\left \lceil n/2^j \right \rceil$ sets $D_{1,j}, \ldots, D_{r,j}$ where $D_i = [i 2^j+1,(i+1)2^j]$. 
\item For each $v \leq r$ each $i \in [10]$ such that $v+i \leq r$  we overlay the DAG $T_{2^j}^{\delta/10}$ on top of $D_{v,j}$ and $D_{v+i,j}$ (Edge Case: if $\left| D_{r,j}\right| = q \leq 2^j$ then we instead overlay $T_{2^j}^{\delta/10} - \{b_{q+1},\ldots,b_{2^j}\}$ on top of $D_{r-i}$ and $D_{r}$ for $i \in [10]$, where $T_{2^j}^{\delta/10}= (V= A \cup B, E) $ and $\{b_{q+1},\ldots,b_{2^j}\}$ denotes the last $2^j-q$ nodes in $B$.  ).  
\end{enumerate}
 By overlaying these expander graphs we can ensure that for {\em any} node $v \in [n]$ of $G_n^\delta$ and any interval $r \geq m_{\delta/10}$ we have the property that for all $X \subseteq [v,v+r-1] Y \subseteq [v+r,v+2r-1]$ such that $|X| \geq \delta r$ and $ |Y| \geq \delta r$ we have $E_n \cap X \times Y \neq \emptyset$ e.g., see \cite[Claim 1]{EGS75}.  Finally, to ensure local expansion between intervals of the form $[v,v+r-1]$ and $[v+r,v+2r-1]$ with $r < m_{\delta/10}$ we can add all edges of the form $\{(i,i+j)~: n\geq i + j \wedge j-i \leq \max\{ m_{\delta/10}, 4 \log n\} \}$. This last step is a modest deviation of \cite{EGS75} since we want to ensure that $G_n^\delta$ is a $\delta$-local expander for all $n > 0 $ and any constant $\delta > 0$. The graph has $\indeg(G_n) \leq 10 c_{\delta} \log n + \max\{ m_{\delta/10}, 4 \log n\} = O(\log n)$. 
\end{proofof}

\begin{remindertheorem}{\thmref{optimalCC}} 
\thmrefOptimalCC
\end{remindertheorem}

\begin{proofof}{\thmref{optimalCC}}
We set $\epsilon = \eta^2/100$ and let $G_n^\eta$ be the graph $G_n^\epsilon$ from the proof of \thmref{ExtremeDepthRobust}. In particular, $G_n^\eta$ is a $\delta=\epsilon/10$-local expander and we set $\gamma=\epsilon/4$ when we consider $\gamma$-good nodes.

Consider a legal pebbling $P \in \pPeb_{G_n^\eta}$ and let $t_i$ denote the first time that node $i$ is pebbled ($i \in P_{t_i}$, but $i \notin \bigcup_{j < t_i} P_j$). We consider two cases:
\begin{enumerate}
\item[Case 1] $\left|P_{t_i}\right| \geq \left( 1-\eta/2\right) i$. Observe that if this held for all $i$ then we immediately have $\sum_{j=1}^t \left|P_i \right| \geq \sum_{j=1}^n \left| P_{t_i} \right| \geq \left(1-\eta/2\right) \sum_{i=1}^n i \geq \frac{n^2\left(1-\epsilon/2\right)}{2}$.
\item[Case 2] $P_{t_i} < \left(1-\eta/2\right) i$. Let $GOOD_i$ denote the set of $\gamma$-good nodes in $[i]$. We observe that at least $i-(1-\eta/2) i\frac{1-\gamma}{1+\gamma} \geq i \eta/4$ of the nodes in $[i]$ are $\gamma$-good by \lemmref{CountGoodNodes}. Furthermore, we note that the subgraph $H_i = G_n^\eta[GOOD_i]$ is  $\left(a\left| Good_i \right|,\left(1-a\right) \left| Good_i \right|-\epsilon i\right)$-depth robust for any constants $a>0$. \footnote{To see this observe that if $G_n^\epsilon$ is a $\delta$-local expander then $G_n^\epsilon[i]$ is also a $\delta$-local expander. 
Therefore, \lemmref{GoodNodesConnected} and \lemmref{CountGoodNodes} imply that $G_n^\epsilon[i]$ $(ai,bi)$-depth robust for any $a+b \leq 1-\epsilon$. Since, $H_i$ is a subgraph of $G_n^\epsilon[i]$ it must be that $H_i$ is $\left(a\left|Good_i \right|,\left(1-a\right) \left| Good_i\right|-\epsilon i\right)$-depth robust. Otherwise, we have a set $S \subseteq V(H_i)$ of size $a\left| Good_i \right|$ such that $\depth(H_i-S) < \left(1-a\right) \left| Good_i \right|-\epsilon i$ which implies that $\depth(G_n^\epsilon[i] - S) \leq i-|Good_i| + \depth(Good_i-S) < i -a|Good_i|-\epsilon i$ contradicting the depth-robustness of  $G_n^\epsilon[i]$.}   

Thus, a result of Alwen et al.~\cite{ABP17} gives us $\ppeb_\Ccc\left(H_i\right)  \geq i^2 \eta^2/100$ since the DAG $H_i$ is at least $\left(i \eta/10, i \eta/10\right)$-depth robust. To see this set $a=1/2$ and observe that $a|Good_i| \geq i\eta/8$ and that $\left(1-a\right) \left| Good_i \right|-\epsilon i \geq i\eta/8 - \eta i/100 \geq  i\eta/10$. Similarly, we note that at time $t_i$ the node $i+\gamma i$ is $\gamma$-good. Thus, by \lemmref{GoodNodesConnected} we will have to completely repebble $H_i$ by time $t_{i+\gamma i}$. This means that $\sum_{j=t_i}^{t_{i+\gamma i}} \left|P_j\right| \geq \ppeb_\Ccc\left(H_i\right)  \geq i^2 \eta^2/100$ and, since $\gamma = \eta^2/400$ we have  $i^2 \eta^2/100 > 2 \gamma i^2 > \sum_{j=i}^{i+\gamma i} j (1-\eta/2)$ . 
\end{enumerate}
Let $x_1$ denote the first node $1 \leq x_1 \leq n-\gamma n$ for which $\left| P_{t_{x_1}}\right| < \left(1-\eta/2\right) i$ and, once $x_1,\ldots,x_k$ have been defined let $x_{k+1}$ denote the first node such that $n-\gamma n > x_{k+1} > \gamma x_k + x_k$ and $\left| P_{t_{x_{k+1}}}\right| < \left(1-\eta/2\right) i$. Let $x_1,\ldots,x_{k*}$ denote a maximal such sequence and let $F = \bigcup_{j=1}^{k*} [x_j,x_j+\gamma x_j]$. Let $R=[n-\gamma n]\setminus F$. We have $\sum_{j \in R} \left| P_j \right| \geq \sum_{j \in R} j(1-\eta/2)$ and we have  $\sum_{j \in F}  \left| P_j \right| \geq \sum_{j \in R} j(1-\eta/2)$. Thus,  \[\sum_{j=1}^t \left|P_i \right| \geq \sum_{j \in R} \left| P_j \right|+ \sum_{j \in F} \left| P_j \right|  \geq \sum_{j= 1}^{n-\gamma n} \frac{n^2\left(1-\eta/2\right)}{2} \geq \frac{n^2\left( 1-\eta/2 \right)}{2}  - \gamma n^2 \geq \frac{n^2\left( 1-\eta \right)}{2} \ .\]
\end{proofof}

\begin{definition}[Parallel White Sequential Graph Pebbling]\deflab{pebbling}
Let $G= (V,E)$ be a DAG and let $T \subseteq V$ be a target set of nodes to be pebbled. 
A {\em black-white pebbling configuration (of $G$)} consists of two subset $P_i^B, P_i^W \subseteq V$. A legal \emph{parallel} pebbling of $T$ is a sequence $P=(P_0,\ldots,P_t)$ of {\em black-white pebbling configurations} of $G$ where $P_0 = (\emptyset,\emptyset)$ and which satisfies the following conditions: (1) the last pebbling configuration contains no white pebbles i.e. $P_t = (P_t^B, P_t^W)$ where $P_t^W = \emptyset$, (2) at most one white pebble is placed per step i.e. $ \forall i \in [t]~~:~~ |P_i^W \setminus P_{i-1}^W|\le 1$, (3) a white pebble can only be removed from a node if all of its parents were pebbled at the end of the previous step i.e., $ \forall i \in [t]~~:~~ x \in (P_{i-1}^W \setminus P_{i}^W) ~\Rightarrow~ \parents(x) \subseteq P_{i-1}^W \cup P_{i-1}^B$, (4) a black pebble can only be added if all its parents were pebbled at the end of the end of the previous step i.e.,   $ \forall i \in [t]~~:~~ x \in (P_i^B \setminus P_{i-1}^B) ~\Rightarrow~ \parents(x) \subseteq P_{i-1}^W \cup P_{i-1}^B$, (5) at some step every node is pebbled (though not necessarily simultaneously) i.e.,  $ \forall x \in T~ \exists z \le t ~~:~~x\in P_z^W \cup P_z^B$.
We denote with $\Peb_{G}^{BW}$ the set of all parallel-black sequential white pebblings of $G$. We use $\peb^{BW}_{cc}(G) = \min_{P \in \Peb_{G}^{BW}} \peb_{cc}^{BW}(P)$ where for $P=(P_0,\ldots,P_t)$ we have $\peb_{cc}^{BW}(P) = \sum_{i=1}^t \left| P_i^B \cup P_i^W \right|$. 
\end{definition}

\begin{remindertheorem}{\thmref{whitePebbleUpper} }
\thmWhitePebbleUpper
\end{remindertheorem}

\begin{proofof}{\thmref{whitePebbleUpper} } Let $S = \{v_1,\ldots,v_e\} \subseteq V$ be given such that $\depth(G-S) \leq d$. For pebbling rounds $i \leq e$ we set $P_i^W = v_i \cup P_i^W$ and $P_i^B= \emptyset$. For pebbling rounds $e < i \leq e+d$ we set $P_i^W = P_{i-1}^W$ and $P_i^B = P_{i-1}^B \cup \{ x :   \parents(x) \subseteq P_{i-1}^W \cup P_{i-1}^B\}$ so that $P_i^B$ contains every node that can be legally pebbled with a black pebble. Finally, we set $P_{e+d+1} = (\emptyset,\emptyset)$. Clearly, the cost of this pebbling is at most
\[\peb^{BW}_{cc}(G) \leq dn + \sum_{i=1}^e i   = dn+ \frac{e(e+1)}{2}  \ , \]
since $\left|P_i^W \cup P_i^B \right| = i$ for $i \leq e$, $\left| P_{e+d+1}^W \cup P_{e+d+1}^B \right| =0$  and we always have $\left|P_i^W \cup P_i^B \right| \leq n$ during any other round $i$. We now show that the proposed pebbling is legal. We only remove white pebbles during the last round $e+d+1$ so rule (3) is trivially satisfied for rounds $i \leq e+d$. We claim that $P_{e+d}^W \cup P_{e+d}^B = V$. Observe that if this claim is true then rule (5) is satisfied and the last pebbling configuration satisfies rule (3). By definition, the last configuration $P_{e+d+1} = (\emptyset,\emptyset)$ contains no white pebbles so rule (1) is satisfied. Clearly, rounds $i \leq e$ are legal with respect to rules (2) and (4) since we place at most one new white pebble on the graph at each point in time. Similarly, during rounds $e+1,\ldots,e+d$ we don't add/remove white pebbles and $P_i^B$ is defined to only include nodes on which a black pebble can be legally pebbled. Thus, rules (2) and (4) are satisfied during all rounds.  

It remains to verify that $P_{e+d}^W \cup P_{e+d}^B = V$. To see this we note that at round $e$ we have $\depth(G-(P_e^W \cup P_e^B)) = \depth(G-S) \leq d$. We now observe that during each subsequent round the depth is reduced by $1$ i.e.,  for $e < i \leq d$ we have $\depth(G-(P_i^W \cup P_i^B)) \leq \depth(G- (P_{i-1}^W \cup P_{i-1}^B) -1$. It follows that $\depth(G-(P_{e+d}^W \cup P_{e+d}^B)) \leq 0$, which can only be true if $P_{e+d}^W \cup P_{e+d}^B = V$.

To validate the last claim we simply observe that any DAG $G$ with $\indeg(G)=O(1)$ is $(e,d)$-reducible with $e= O\left(\frac{n \log \log n}{\log n}  \right)$ and $d= O\left(\frac{n}{\log^2 n} \right)$~\cite{AB16}. 
\end{proofof}

\begin{remindertheorem}{\thmref{HighCCBW}} 
\ThmHighCCBW
\end{remindertheorem}

\begin{proofof}{\thmref{HighCCBW}} 
We first show that $\peb^{BW}_{cc}(G) \geq  \left(1/16-\epsilon/2  \right)n^2 $. Let $P = \left(P_0, P_1,\ldots,P_t\right) \in \Peb{G}^{BW}$ be given and for simplicity assume that $n/4$ is an integer. Let $B_i = \bigcup_{j \leq 4t/n} \left( P_{i+jn/4}^W \cup P_{i+jn/4}^B \right)$ for $i \in [n/4]$.  We claim that for each $i$ we have $\left| B_i \right| \geq  (3/4-2\epsilon) n$. If this holds then we have $\left(3/16-\epsilon/2  \right)n^2 \leq  \sum_{i \in [n/4]} \left|B_i \right| \leq \sum_{i \in [t]} \left|P_i^W \cup P_i^B \right|$, and the final claim will follow immediately from \thmref{ExtremeDepthRobust}..  

It remains to verify our claim. Consider the interval $[i+jn/4+1,i+(j+1)n/4-1]$ for some arbitrary $j$ and let $S =\bigcup_{r=i+jn/4+1}^{i+(j+1)n/4-1} P_r^W \setminus  P_{i+jn/4}^W$ denote the set of white pebbles placed on $G-B_i$ during this interval. Let $H = \ancestors_{G-B_i}(S)$. Because all white pebbles placed on $S$ were removed by round $i+jn/4$ we note that $H \subseteq \bigcup_{r=i+jn/4+1}^{i+(j+1)n/4-1} \left( P_r^W \cup P_r^B\right)$. Since, $H$ must have been pebbled completely during the interval this means that $\depth(H-S) \leq n/4$ since we never place white pebbles on nodes in $V(H-S)$. Thus, $H$ is $(n/4,n/4)$-reducible. On the other hand we note that, by depth-robustness of $G$, $H$ must be $(e,d)$-depth-robust for any $(e,d)$ such that $e+d \leq \left| V_H\right| - \epsilon n$. It follows that $\left| V_H\right|  \leq n(1/2+\epsilon)$. For any node $x \in V(G-B_i)$ that is pebbled during the interval $[i+jn/4+1,i+(j+1)n/4-1]$ the length of the longest path to $x$ in $G-B_i$ can be at most $\depth(V_H) + n/4 \leq \left| V_H\right| + n/4 \leq n(1/2+\epsilon)+n/4$. Thus, we have $\depth(G-B_i) \leq 3n/4+\epsilon n$. Since, $G$ is $(e,d)$-depth-robust for any $e+d \leq (1-\epsilon)n$ we must have $\left| B_i \right| \geq (1-\epsilon)n - 3n/4 - \epsilon n = n/4 -2\epsilon n$.

A similar argument shows that $\peb^{BW}_{cc}(G') \geq  \left(1/16-\epsilon/2  \right)n^2$. Since the argument requires some adaptations we repeat it below for completeness. 

Let $P = \left(P_0, P_1,\ldots,P_t\right) \in \Peb{G'}^{BW}$ be given and for simplicity assume that $n/4$ is an integer. Let $B_i' = \bigcup_{j \leq 4t/n} \left( P_{i+jn/4}^W \cup P_{i+jn/4}^B \right)$ for $i \in [n/4]$ and let $B_i = \{v \in [n] : [2v\d+1,2(v+1)\d] \cap B_i' \neq \emptyset \}$ be the corresponding nodes in original DAG $G$.  We claim that for each $i$ we have $\left| B_i' \right| \geq  (3/4-2\epsilon) n$. If this holds then we have 
\[ \left(3/16-\epsilon/2  \right)n^2 \leq  \sum_{i \in [n/4]} \left|B_i' \right|  \leq  \sum_{i \in [t]} \left|P_i^W \cup P_i^B \right| \ ,\]
so that $\peb^{BW}_{cc}(G') \geq  \left(1/16-\epsilon/2  \right)n^2$. The theorem follows immediately from \thmref{ExtremeDepthRobust} we can take $G$ to be an $(e,d)$-depth-robust DAG on $n$ nodes with $\indeg(G) = O(\log n)$. $G'$ is now an $n'=2n\d = O(n \log n)$ node DAG with  $\peb^{BW}_{cc}(G') = \Omega(n^2) = \Omega(n^2/\log^2 n)$. 

To verify that our claim holds consider the interval $[i+jn/4+1,i+(j+1)n/4-1]$ for some arbitrary $j$ and let $S =\bigcup_{r=i+jn/4+1}^{i+(j+1)n/4-1} P_r^W \setminus  P_{i+jn/4}^W$ denote the set of white pebbles placed on $G'-B_i'$ during this interval. Let $H' = \ancestors_{G'-B_i'}(S)$. Because all white pebbles placed on $S$ were removed by round $i+jm$ we note that $H' \subseteq \bigcup_{r=i+jn/4+1}^{i+(j+1)n/4-1} \left( P_r^W \cup P_r^B\right)$. Since, $H'$ must have been pebbled completely during the interval this means that $\depth(H'-S) \leq n/4$ since we never place white pebbles on nodes in $V(H'-S)$. Thus, $H'$ is $(n/4,n/4)$-reducible. On the other hand we consider the graph $H = \ancestors_{G-B_i}\left( \{ v: S \cap [2v\d+1,2(v+1)\d] \neq \emptyset \}\right)$. By depth-robustness of $G$, $H$ must be $(e,d)$-depth-robust for any $(e,d)$ such that $e+d \leq \left| V_H\right| - \epsilon n$. It follows that $\left| V_H\right|  \leq n(1/2+\epsilon)$. Furthermore, $H'$ is a subgraph of the indegree reduced version of $H$ so $\left| V_{H'}\right| \leq \d \left| V_H \right| \leq \d  n(1/2+\epsilon)$. 

 For any node $x \in V(G'-B_i')$ that is pebbled during the interval $[i+jn/4+1,i+(j+1)n/4-1]$ the length of the longest path to $x$ in $G'-B_i'$ can be at most $\depth(V_{H'}) + n/4 \leq \left| V_{H'}\right| + n/4 \leq \d n(1/2+\epsilon)+n/4$. Thus, we have $\depth(G'-B_i') \leq \d n (1/2+\epsilon) + n/4$.  Since, $G$ is $(e,d)$-depth robust for any $e+d \leq (1-\epsilon)n$ it follows from \lemmref{IndegRed} that $G'$ is $(e,d\d)$-depth-robust for any $e+d \leq (1-\epsilon)n$~\cite{ABP17}. Therefore, we have $\left| B_i' \right| \geq (1-\epsilon)n - n/2 - \epsilon n - n/(4\d) \geq n/4 -2\epsilon n$.
\end{proofof}

\begin{theorem} \cite{PTC76} \thmlab{PTC}
There is a family of DAGs $\{G_n\}_{n=1}^\infty$ with $\indeg\left(G_n\right) = 2$ with the property that for some positive constants $c_1,c_2,c_3> 0$  such that for each $n \geq 1$  the set $S = \{ v : \parents(G_n) = \emptyset\}$ of sources has size $\left| S\right| \leq c_1n/\log n$ and for any legal pebbling $P=\left(P_1,\ldots,P_t\right) \in \Peb (G_n)$ there is an interval $[i,j] \subseteq [t]$ during which at least $c_2 n/\log n$ nodes in $S$ are (re)pebbled (formally, $\left| S \cap \bigcup_{k =i}^j P_k - P_{i-1} \right| \geq c_2n/\log n$)  and at least $c_3n/\log n$ pebbles are always on the graph ($\forall k \in [i,j], \left| P_k \right| \geq c_3n/\log n$). 
\end{theorem}

\begin{remindertheorem}{\thmref{PTCParallel}}
\thmPTCParallel
\end{remindertheorem}

\begin{proofof}{\thmref{PTCParallel}}
The family of DAGs $\{G_n\}_{n=1}^\infty$ is the same as in \thmref{PTC}. Similarly, let $c_1,c_2,c_3 > 0$ denote the constants from \thmref{PTC} and let $S \subseteq V$ be the set of $\left| S\right| \leq c_1n/\log n$ nodes from \thmref{PTC}. 

Let $P=\left(P_1,\ldots,P_t\right) \in \pPeb_{G_n}$ be any pebbling of $G_n$. We consider the sequential transform $P' = \seq(P) \in \Peb_{G_n}$ from \defref{seq}. Recall that $P'=\left(P_1',\ldots,P_{A_t} \right)$ where $A_k = \sum_{i=1}^k \left| P_i \setminus P_{i-1}\right|$ (and $P_0 \doteq \emptyset$).  By \lemmref{seqparspace} $P'$ is a legal sequential pebbling $P' \in \Peb(G_n)$. Furthermore, we note that $P_{A_i}' = P_i$ for all $i \leq t$ and that $P_i \subset P_{A_i+k}' \subseteq P_{i} \cup P_{i+1}$ for each $i \leq k$ and $k \leq \left|P_{i+1}\setminus P_i\right|$.  

Let $t_2^*$ denote the maximum value such that there exists an interval $[t_1^*,t_2^*] \subseteq [A_t]$ such that 

\begin{equation} \label{cond1}
\left| S \cap \bigcup_{k =t_1^*}^{t_2^*} P_k' - P_{k-1}' \right| \geq c_2n/\log n\ , \end{equation}
 and \begin{equation} \label{cond2} \forall k \in [t_1^*,t_2^*], \left| P_k' \right| \geq c_3n/\log n \ .\end{equation} Observe that by \thmref{PTC} $t_2^*$ must exist. Having fixed $t_2^*$ let $t_1^* < t_2^*$ denote the minimum value such that the above properties hold for the interval $[t_1^*,t_2^*]$. 

We first claim that $t_2^* = A_j$ for some $j \leq t$. Suppose instead that $t_2^* = A_j+k$ for $0 < k < \left|P_{j+1}\setminus P_j\right|$. In this case, we have $P_{t_2^*}' \subsetneq P_{t_2^*+1}'$ which implies that $\left|P_{t_2^*+1}' \right| \geq \left| P_{t_2^*}' \right| \geq c_3'n/\log n$. Furthermore, $\left| S \cap \bigcup_{k =t_1^*}^{t_2^*+1} P_k' - P_{k-1}' \right| \geq  \left| S \cap \bigcup_{k =t_1^*}^{t_2^*} P_k' - P_{k-1}' \right| \geq c_2'n/\log n$ so the interval $[t_1^*,t_2^*+1]$ satisfies conditions \ref{cond1} and \ref{cond2} above. This contradicts the minimality of $t_2^*$.  

Now suppose that $t_1^* = A_i+k$ for some $0 \leq i \leq t$ and $a_{i+1} > k \geq 0$ and consider the interval $[i+\mathbbm{1}_{k>0},j]$.  

If $k > 0$ we have  the interval $[i+1,j]$. We note that $t_1^* < t_2^*$ and thus $i+1 \leq j$. Now
\[\left| S \cap \bigcup_{x =i+1}^{j} P_x - P_{i} \right|=\left| S \cap \bigcup_{x =A_{i+1}}^{A_j} P_x' - P_{A_i}' \right| \geq \left| S \cap \bigcup_{x =t_1^*}^{t_2^*} P_x' - P_{t_1^*-1}' \right| \geq c_2'n/\log n \  \]
where the second to last inequality follows because $P_{t_1^*-1} = \bigcup_{x=A_i}^{t_1^*-1} P_x$. Furthermore, for each $i < x \leq j$ we know that $\left|P_x\right| = \left|P_{A_x}\right| \geq c_3n/\log n$ since $t_1^* \leq A_x \leq t_2^*$. Thus, the interval $[i+1,j] \subseteq [t]$ satisfies both required properties. 

If instead $k = 0$ we have the interval $[i,j]$. In this case 
\[\left| S \cap \bigcup_{x =i}^{j} P_x - P_{i-1} \right|=\left| S \cap \bigcup_{x =A_{i}}^{A_j} P_x' - P_{A_{i-1}}' \right| =\left| S \cap \bigcup_{x =t_1^*}^{t_2^*} P_x' - P_{A_{i-1}}' \right| \geq \left| S \cap \bigcup_{x =t_1^*}^{t_2^*} P_x' - P_{t_1^*-1}' \right| \geq c_2'n/\log n \ , \]
where the second to last inequality follows since $P_{t_1^*-1}'=P_{A_{i-1} + a_i-1} \supset P_{A_{i-1}}$.  Furthermore, for each $i \leq  x \leq j$ we know that $\left|P_x\right| = \left|P_{A_x}\right| \geq c_3n/\log n$ since $t_1^* \leq A_x \leq t_2^*$. \QED

\end{proofof}

\begin{claim} \claimlab{LDImpliesSSOne}
Let $G_n^\epsilon$ be an DAG with nodes $V\left(G_n^\epsilon \right)=[n]$, indegree $\d = \indeg\left( G_n^\epsilon\right)$ that is $(an,bn)$-depth robust for all constants $a,b>0$ such that $a+b\leq 1-\epsilon$, let $G$ be the indegree reduced version of $G_n^\epsilon$ from \lemmref{IndegRed} with nodes and $\indeg(G) = 2$ and let $P=(P_1,\ldots,P_t) \in \pPeb_G$ be a legal pebbling of $G$ such that during some round $i$ the length of the longest unpebbled path in $G$ is at most $\depth\left(G-P_i\right) \leq c\d n$ for some constant $1 > c>0$. Then $\peb_\Css\left(P, n(1-\epsilon-2c) \right) \geq c \d n$.
\end{claim}
\begin{proof}
Let $k < i$ be the last pebbling step before $i$ during which the length of the longest unpebbled path at time $k$ is at most  $\depth(G-P_k) \geq  2cn\d$. Observe that $k-i \geq \depth(G-P_k) - \depth(G-P_i) \geq cn\d$ since we can only decrease the depth by at most one in each pebbling round. In particular, $\depth(G-P_k) = 2c\d$ since $1+\depth(G-P_{k}) \leq  \depth(G-P_{k+1}) < 2cn\d$. Let $r \in [k,i]$ be given then by construction we have $\depth\left(G- P_r\right) \leq 2cn\d$. Let $P_r' = \{v \in V(G_n^\epsilon): P_r \cap [2\d(v-1)+1,2\d v] \neq \emptyset \}$ be the set of nodes $v$ in $G_n^\epsilon$ such that the corresponding path $2\d(v-1)+1,\ldots,2\d v$ contains no pebble at time $r$. Exploiting the properties of the indegree reduction from \lemmref{IndegRed} we have \[\depth\left(G_n^\epsilon - P_r'\right) \d \leq \depth\left(G - \bigcup_{v \in P_r'} [2\d(v-1)+1,2\d v]\right) \leq \depth\left(G - P_r\right) \leq 2cn\d \ .   \]
 Now by depth-robustness of $G_n^\epsilon$  we have 
\[ \left| P_r' \right| \geq (1-\epsilon)n - \depth\left(G_n^\epsilon- P_r'\right) \geq n-\epsilon n - 2cn \ .  \]
Thus,  $\left| P_r\right| \geq \left| P_r' \right| \geq n(1-\epsilon-2c)$ for each $r \in [k,i]$. It follows that  $\peb_\Css\left(P, n(1-\epsilon-2c) \right) \geq c \d n$.
\end{proof}

\begin{reminderlemma}{\lemmref{IndegRed}}\cite[Lemma 1]{ABP17}  
\indegReductionLemma
\end{reminderlemma}
The proof of \lemmref{IndegRed} is essentially the same as \cite[Lemma 1]{ABP17}. We include it here is the appendix for completeness. 

\begin{proofof}{\lemmref{IndegRed}}
We identify each node in $V'$ with an element of the set $V\times[2\d]$ and we write $\node{v}{j} \in V'$. For every node $v\in V$ with $\a_v := \indeg(v) \in [0,\d]$ we add the path $p_v=(\node{v}{1}, \node{v}{2}, \ldots,\node{v}{2\d})$ of length $2\d$. We call $v$ the {\em genesis node} and $p_v$ its {\em metanode}. In particular $V'= \cup_{v\in V} p_v$. Thus $G$ has size at most $(2\d)n$.

Next we add the remaining edges. Intuitively, for the $i$\th incoming edge $(u,v)$ of $v$ we add an edge to $G'$ connecting the end of the metanode of $u$ to the $i$\th node in the metanode of $v$. More precisely, for every $v\in V$, $i\in[\indeg(v)]$ and edge $(u_i,v)\in E$ we add edge $(\node{u_i}{2\d}, \node{v}{i})$ to $E'$. It follows immediately that $G'$ has indegree (at most) $2$. 

Fix any node set $S'\subset V'$ of size $|S'| \le e$. Then at most $e$ metanodes can share a node with $S'$. Let $S = \{v~:~\exists j \in [2\d]~\mathbf{s.t.} \node{v}{j} \in S' \}$ denote the set of genesis nodes in $G$ whose metanode shares a node with $S'$ and observe that $|S| \leq |S'|$.  For each such metanode remove its genesis node in $G$. Let $p = (v_1,\ldots, v_k)$ be a path in $G-S$. After removing nodes $S'$ from $G'$ there must remain a corresponding path $p'$ in $G'$ running through all the metanodes of $p$ and $|p'| \geq |p|\d$ since for each $v_j$ $p'$ at minimum contains the nodes $\node{v_j}{\d},\ldots,\node{v_j}{2\d}$. In particular, $G'$ must be $(e,d \d)$-depth robust.


\end{proofof}

\subsection{Proof of Pebbling Reduction}\applab{pebproof}

\begin{remindertheorem}{\thmref{pebred}}
\PebblingReductionThm
\end{remindertheorem}

\begin{proofof}{\thmref{pebred}}[Sketch]
 We begin by describing the simulator $\simul$ for $\rv = (q,s,t)$. Recall that it can make up to $\b(\rv)$ calls to $f_n^{(h')}$ (where $h'\from \Hc$ is uniform random). Essentially $\simul$ runs a copy of algorithm $\A$ on an emulated PROM device parametrized by resource bounds $\rv$. For this $\simul$ emulates a RO $h \in \Hc$ to $\A$ as follows. All calls to $h$ are answered consitently with past calls. If the query $\bar{x}$ has not previously been made then $\simul$ checks if it has the form $\bar{x}=(x,u,\l_1,\l_2)$ where all of the following three conditions are met:
 \begin{enumerate}
    \item
    $u = v_{\out}$ is the sink of $G$,
    
    \item
    $\l_1$ and $\l_2$ are the labels of the parents of $u$ in $G$ in the $(h,x)$-labeling of $G$,
    
    \item
    $\A$ has already made all other calls to $h$ for the $(h,x)$-labeling of $G$ in an order respecting the topological sorting of $G$.
 \end{enumerate}
 \noindent
 We call a query to $h$, for which the first two conditions are valid, an \emph{$h$-final call (for $x$)}. If the third condition also holds then we call the query a \emph{sound} final call. Upon such a fresh final call $\simul$ forwards $x$ to $f_n^{(h')}$ to obtain response $y$. It records $(\bar{x},y)$ in the function table of $h$ and returns $x$ to $\A$ as the response to its query. If the response from $f_n^{(h')}$ is $\bot$ (because $\simul$ has already made $\b(rv)$ queries) then $\simul$ outputs $\bot$ to the distinguisher $\dist$ and halts. We must show that $\dist$ can not tell an interaction with such an ideal world apart from the real one with greater than probability $\eps(\lv,\rv,n)$.

 Next we generalize the pebbling game and notion of sustained space to capture the setting where multiple identical copies of a DAG $G$ are being pebbled. In particular, in our case, when $P=(P_0, P_1,\ldots)$ is a pebbling of $m$ copies of $G$ then we define the {\em $s$-block memory} complexity is defined to be $\ppeb_\Cbm(P) = \sum \flr{|P_i|/s}$. It follows immediatly that $G_m$ consists of $m$ independent copies of $G$ then $\ppeb_\Cbm(G_m) \ge m*\ppeb_\Css(G)$.
 
 The next step in the proof describes a mapping between executions of a pROM algorithm $\A$ and a pebbling of multiple copies of $G$ called the {\em ex-post-facto} pebbling of the execution. This technique was first used in~\cite{DNW05} and has been used in several other pebbling reductions~\cite{DKW11,AS15,AT17}. For our case the mapping is identical to that of~\cite{AT17} as are the following two key claims. The first states that with high probability (over the choice of coins for $\A$ and choice of the random oracle $h$) if $\A$ computed $m$ outputs (of distinct inputs) for $f_n^{(h)}$ then ex-post-facto pebbling of that execution will be a legal and complete pebbling of $m$ copies of $G$. The second claim goes as follows.
 
 \begin{claim}\claimlab{bounded}
  Fix any input $x_{\inp}$. Let $\s_i$ be the i\th input state in an execution of $\A^h(x_{\inp};\coins)$. Then, for all $\lambda \ge 0$,
  $$\Pr\left[\forall i ~:~ \sum_{x \in X} \abs{P^x_i} \le \frac{\abs{\s_i} + \lambda}{w - \log(q_r)}\right] > 1 - 2^{-\lambda}$$
  over the choice of $h$ and $\coins$.
 \end{claim}

 In particular this implies that the size of each individual state in the pROM execution can be upper-bounded by the number of pebbles in the corresponding ex-post-facto pebbling. More generally, the $s$-block memory complexity of the ex-post-facto pebbling gives us an lower-bound on the $s$-SMC of the execution. Since the block memory complexity of a graph can be lowerbounded by the sustained space complexity of the graph these results lead to a lowerbound on $s$-sustained memory complexity of the graph function in terms of the $s$-sustained space complexity of $G$.
\end{proofof}